\pdfoutput=1 
\documentclass[11pt,a4paper]{article}
\usepackage[a4paper,margin=1in]{geometry}

\usepackage{cite}
\usepackage{xcolor}

\definecolor{darkPurple}{HTML}{3333B2}
\definecolor{forestGreen}{HTML}{337700}
\definecolor{shadecolor}{HTML}{CCCCCC}
\definecolor{darkGrey}{HTML}{4E4F86}
\definecolor{Crimson}{HTML}{DC143C}
\definecolor{transparentBlue}{HTML}{EBEBF8}

\usepackage{hyperref}
\hypersetup{colorlinks=true, linkcolor=darkPurple, citecolor=darkPurple}

\usepackage[pdftex]{graphicx}
\usepackage{caption}
\usepackage{subcaption}

\graphicspath{{./figures/}}
\usepackage[cmex10]{amsmath}
\usepackage{amssymb}
\usepackage{multirow}	
\usepackage[toc]{appendix}

\usepackage{amsthm}
\usepackage{commath} 
\newtheorem{theorem}{Theorem}
\newtheorem{definition}{Definition}

\providecommand{\ceil}[1]{\left \lceil #1 \right \rceil }
\providecommand{\floor}[1]{\left \lfloor #1 \right \rfloor }
\providecommand{\abs}[1]{\left \vert  #1 \right \vert }

\newenvironment{denseItemize}
{ \begin{itemize}
    \setlength{\itemsep}{0pt}
    \setlength{\parskip}{0pt}
    \setlength{\parsep}{0pt}     }
{ \end{itemize}                  }

\providecommand{\keywords}[1]{\textbf{\textit{Keywords---}}\small{\textit{#1}}}
\begin{document}
\title{On Continuous-space Embedding of Discrete-parameter Queueing Systems}
\date{}
\author{
	Neha Karanjkar \footnote{Post-doctoral Fellow, Robert Bosch Centre for Cyber-physical Systems, IISc Bangalore}
	\qquad
	Madhav P. Desai \footnote{Professor, Department of Electrical Engineering, IIT Bombay}
	\qquad
	Shalabh Bhatnagar \footnote{Professor and Chair, Department of Computer Science and Automation, IISc Bangalore}
	\\[2em]
}
\maketitle
\vspace{-1em}


\begin{abstract}
	
	Motivated by the problem of discrete-parameter simulation optimization
	(DPSO) of queueing systems, we consider the problem of embedding the
	discrete parameter space into a continuous one so that descent-based
	continuous-space methods could be directly applied for efficient optimization.
	We show that a randomization of the simulation model itself can be used
	to achieve such an embedding when the objective function is a
	long-run average measure.
	Unlike spatial interpolation, the computational cost of this embedding
	is independent of the number of parameters in the system, making the
	approach ideally suited to high-dimensional problems.
	We describe in detail the application of this technique 
	to discrete-time queues for embedding queue capacities, number of servers
	and server-delay parameters into continuous space and empirically
	show that the technique can produce smooth interpolations of the objective function.
	Through an optimization case-study of a queueing network with $10^7$ 
	design points, we demonstrate that existing continuous optimizers can be
	effectively applied over such an embedding to find good solutions.
	\\

\end{abstract}

\newpage

\section{Introduction}
	
	\subsection{Motivation}
	The use of simulation is often necessary in the optimization of complex
	real-life queueing networks. Such queueing networks typically have
	discrete valued parameters such as queue capacities, the
	number of servers and server-delays.
	More concretely, consider a queuing system with a parameter set 
	$X=(x_1,x_2,\dots,x_n)$ where each $x_i$ can take 
	integer values between some fixed bounds.
	The set of all possible values that $X$ can take is the 
	$n$-dimensional discrete parameter space
	$\Omega_D$.
	Let $f:\Omega_D\rightarrow\mathbb{R}$ be a cost/performance 
	measure of the system that we wish to optimize.
	In many applications, $f$ may be composed of long-run average
	measures such as average throughput, average waiting
	time per customer or blocking probabilities. An analytical expression
	for $f$ is rarely available and given $X$, $f(X)$
	can only be measured using a simulation of the system.  The
	measurements are noisy as each simulation run has finite length.
	We are motivated by the problem of finding an $X^*\in\Omega_D$ that
	minimizes $f(X^*)$.  This is a discrete-parameter simulation
	optimization (DPSO) problem. 
	The problem is often difficult as the number of parameters
	can be very large and each function evaluation is computationally expensive.

	For small parameter sets, techniques such as
	Optimal Computing Budget Allocation (OCBA)\cite{Chen_OCBA} 
	have been effectively applied.
	When the number of parameters is large, an exhaustive evaluation
	of all design points is infeasible. The goal then is to find
	the best possible solution within a finite computational budget,
	rather than the global optimum.
	In such a case, randomized search techniques 
	such as simulated annealing\cite{Kirkpatrick_SimAnneal,Alkhamis_SimAnneal_DPSO} 
	and genetic algorithms\cite{Goldberg_GeneticAlgorithms}
	or heuristic-based local search methods such
	as Tabu search\cite{Glover_TabuSearch} have been employed.
	Discrete-space variants of continuous optimizers
	such as Simultaneous Perturbation Stochastic Approximation (SPSA)\cite{Spall1992} have also been proposed
	wherein the parameter estimate is projected back to the discrete space at each iteration
	\cite{Hill_Discrete_SPSA_2001,Bhatnagar_Discrete_SPSA_2005}.
	In all of the above methods, the objective function evaluation 
	is limited strictly to points in the original  
	discrete domain.

	If the discrete parameter space can be embedded into a 
	larger continuous space by using some form of interpolation, 
	the optimization problem can be solved by directly applying descent-based
	continuous-space methods.
	Let $\Omega_C \subset {\mathbb{R}}^n$ be the convex hull of $\Omega_D$. 
	An embedding of the discrete parameter space into a continuous one
	is essentially an interpolation (say $\hat{f}$) of $f$ defined over $\Omega_C$.
	If $\hat{f}$ can be constructed in a computationally efficient manner
	and has a suitable {\em structure}
	(that is, $\hat{f}$ is continuous, piece-wise smooth, has few local minima) then
	continuous optimizers applied directly over $\hat{f}$
	can be expected to perform well. This is because gradient
	information can be utilized at each step to converge
	rapidly to local minima. Random multi-starts can be
	used in the case of non-convex functions. Most importantly,
	such methods often scale well with the number of parameters
	in comparison to enumerative approaches.
	However, the solution needs to be 
	projected back to the original discrete space carefully.
	It is to be noted that an embedding-based approach is essentially different
	from methods where the parameter vector is rounded/truncated
	to integer points at each step. In the former case,
	a continuous-space offers more pathways to reach the 
	solution aggressively.

	The continuous-space embedding approach is suited to
	simulation optimization of queueing, inventory or manufacturing systems
	where most discrete parameters 
	have an integer ordering and the objective function
	has some structure (unlike combinatorial
	problems).
	Such an approach has been recently reported in \cite{Lim2012} 
	and \cite{Wang2008} for the optimization of queueing and inventory
	systems where spatial interpolation 
	(piecewise-local interpolation over a simplex) 
	was used to obtain a continuous-space embedding
	of the objective function.

	\subsection{Problem Statement}
	Motivated by the embedding-based approach to the DPSO problem,
	we consider the problem of obtaining a continuous-space 
	embedding in a {\em computationally efficient} manner.
	Consider a system with the parameter vector $X=(x_1,x_2,\dots,x_n)$
	where each $x_i$ is discrete valued and 
	belongs to the set $\Omega_{D_i}=\{x \mid x\in\mathbb{N}, x^{min}_i \leq x \leq x^{max}_i\}$.
 	Let $\Omega_{C_i}$ denote the range $[x^{min}_i, x^{max}_i]$.
	The set of all possible values that $X$ can take is $\Omega_D = \prod_{i=1}^{n}{\Omega_{D_i}}$
	and $\Omega_C$ is the convex hull of $\Omega_D$.

	We consider the problem of finding $\hat{f}:\Omega_C\rightarrow \mathbb{R}$
	such that $\hat{f}(Y)=f(Y)$ when $Y\in\Omega_D$ and $\hat{f}$ is continuous over $\Omega_C$.
	Further, it is desirable that $\hat{f}$ is 
	continuously differentiable  over $\Omega_C$
	so that gradient-based methods could be applied.
	Spatial interpolation as a means of obtaining 
	$\hat{f}$ is computationally expensive.
	Given $Y\in\Omega_C$ and a set of points
	$X^1, X^2, \dots, X^p \in \Omega_D$ in the neighbourhood of $Y$,
	the interpolation $\hat{f}(Y)$ can be found as:
	$$\hat{f}(Y) = \sum_{k=1}^{p}{a_k f(X^k)}$$ 
	where $a_1,\dots a_p$ are interpolation coefficients. 
	Here $p$ simulations are required to estimate the interpolated value.
	For linear interpolation, $p=2^n$ and for piece-wise simplex interpolation, $p \geq n+1$.
	Thus the computational effort required for the embedding 
	grows rapidly with the dimensionality of the parameter space.

	Instead, we propose an embedding technique which is based on a
	randomization of the simulation model itself.  The interpolated value
	can be computed using a {\em single} simulation of the randomized model,
	irrespective of the dimensionality of the parameter space. 
	The basic idea behind the embedding technique  is described next.
	
	\subsection{Embedding via Randomization}

	Let $Y=(y_1,y_2,\dots,y_n)$ be a point in $\Omega_C$ at which we wish to compute the 
	interpolated value $\hat{f}(Y)$. Thus the $i^{th}$ parameter needs to be assigned
	a value $y_i\in\mathbb{R}$.
	To compute the interpolation, we construct a randomized version of the model
	where each parameter $i$ is perturbed periodically 
	(for example, at the beginning of each time-slot)
	and assigned values of a discrete random variable which we denote as $\gamma_i(y_i)$.
	The random variable $\gamma_i$ is chosen in such a way that
	its moments can be made to vary continuously with respect to the parameter $y_i$,
	and $\gamma_i(y_i)=y_i$ with probability $1$ whenever $y_i\in\Omega_{D_i}$.
	The simplest example of such a random variable is:
	\begin{equation}\label{eq:gamma_example}
	\begin{aligned}
	\gamma(v) &=  \ \left\{ 
	\begin{array}{ll}
	\floor{v} & {\rm with\ probability\ }\alpha(v) \\
	\ceil{v}  & {\rm with\ probability\ } 1-\alpha(v)\\
	\end{array}\right.\\
	{\rm where\ } \alpha(v) &= \ceil{v}-v
	\end{aligned}
	\end{equation}
	Let $\hat{f}$ be the long-run average measure
	obtained by simulating such a randomized model. 
	$\hat{f}$ is now a function of $Y$. 
	Further, $\hat{f}(Y)=f(Y)$ when $Y\in\Omega_D$ by definition.
	Thus $\hat{f}$ is an interpolation of $f$.
	As each parameter in the model can be embedded independently, the
	interpolated value can be computed using a single simulation of the
	randomized model, irrespective of the number of parameters. In essence,
	the interpolation technique relies on averaging in {\em time} in
	contrast to spatial interpolation methods which perform an averaging in
	space. 
	
	Such a technique can be applied in principle, to several types of 
	discrete-event systems where the objective is a long-run average.
	A randomization-based approach was first reported in \cite{ghaza_bottleneck}
	for sensitivity measurement in 	VLSI systems 
	(for producing small real-valued perturbations to discrete-valued
	parameters) and in \cite{KaranjkarMasCOTS} for the 
	design optimization of multi-core system architectures. Both these
	works demonstrated a specific application of the technique,
	however, without a theoretical justification.
	A theoretical basis for the embedding technique 
	was introduced in \cite{T-ASE} in the context of two specific 
	algorithms for solving the DPSO problem. This work proposed
	variants of two continuous optimizers (SPSA and Smoothed Functional (SF) algorithm)
	wherein the parameter estimate at each iteration is projected
	back to the discrete space probabilistically (rather than in a deterministic manner)
	and showed that this effectively produces a continuous embedding of the underlying
	discrete-parameter process. The work focused on the optimization
	algorithms and the embedding technique itself was not explored in depth.

	The focus of the current work is on the randomization-based 
	embedding technique itself and its application to queueing systems
	(rather than a specific optimization method). In fact, once a 
	suitable embedding has been obtained, a rich set of existing
	continuous optimizers become directly applicable 
	to the original DPSO problem.
	The main contributions of this paper are listed below.

	\subsection{Our Contributions}
	
	We present a general scheme for randomization of 
	discrete-valued parameters in a simulation model and show 
	that such a randomization can indeed produce continuous interpolations
	of the long-run average objective.
	The proof (presented in Section \ref{sec:AGeneralRandomizaionScheme})
	uses a result from Markov chain perturbation theory\cite{Schweitzer}
	and builds on a proof in \cite{T-ASE} by relaxing some of its assumptions.
	The interpolation is not unique and can be tuned
	by varying the shape of the weights (probabilities) used during the randomization.
	We refer to these weights as {\em Stochastic Interpolation Coefficients}.
	In Section \ref{sec:GeneratingAlphas} 
	we present a parameterized template for generating these 
	stochastic interpolation coefficients.
	We then discuss the desirable properties of a
	good interpolation and demonstrate how the template parameters
	can be varied to improve the interpolation.

	We explore in detail the application of the embedding
	technique to discrete-time queues.
	Such queues are of importance in many applications 
	such as communication networks, manufacturing lines
	and transportation systems\cite{Alfa2015,bruneel1993}.
	Through several concrete examples, we demonstrate the
	continuous-space embedding of 
	queue-capacity, number of servers and server-delay parameters 
	and observe that the technique can produce smooth
	interpolations of the objective function.
	We present detailed simulation results for these discrete-time queues
	in Section \ref{sec:ApplicationToQueues}.
	The simulation models and scripts used for all
	results reported in this paper are available online (see \cite{Embedding_GitHub}).

	As a demonstration of the utility of this embedding technique,
	we present an optimization case-study of a queueing network
	with $7$ integer parameters and $10^7$ design points.
	We observe that a randomization of the simulation model
	produces smooth embeddings of the objective function
	with a low computational overhead. Two well-established
	continuous-space optimizers (COBYLA\cite{Powell1994} and SPSA\cite{Spall1992}) 
	applied directly over the embedding
	perform well and find good solutions in comparison to a direct 
	discrete-space search. We present the observations from this optimization case-study in 
	Section \ref{sec:OptimizationCaseStudy}. In summary, our work
	establishes a computationally efficient technique for embedding 
	discrete-parameter queueing systems into continuous-space, enabling direct application
	of continuous optimizers to solve the DPSO problem.

\section{A General Randomization Scheme}\label{sec:AGeneralRandomizaionScheme}
	
	We describe the randomization scheme with respect to a single parameter.
	All parameters in the model can be embedded in a similar manner simultaneously
	and independently of each other.
	Consider the $i^{th}$ parameter in the model which can take integer values 
	from the finite set $\Omega_{D_i}= \{x_i^1, x_i^2,\dots,x_i^p\}$.
	Let $\Omega_{C_i}\subset\mathbb{R}$ denote the range $[x_i^1,x_i^p]$.

	To embed the $i^{th}$ parameter into a continuous space and assign to it 
	the value $y_i\in\Omega_{C_i}$,
	we construct a randomized version of the model where the parameter is perturbed
	at multiple time-instants within a single simulation trajectory
	and assigned values of the discrete random variable $\gamma_i(y_i)$.
	The random variable $\gamma_i$ can have 
	a multi-point distribution taking values from the set $\Omega_{D i}$.
	We choose a set of functions $\alpha_i^1, \alpha_i^2, \dots,\alpha_i^p$ which map
	values from the domain $\Omega_{C_i}$ to the range $[0,1]$ such that:
	\begin{itemize}
	
		\item $\sum_{k=1}^{p}{\alpha_i^k(y)}=1 \ \ \ \  \forall y\in\Omega_{C_i}$,
		\hfill\refstepcounter{equation}\textup{(\theequation)}
		\label{eq:alpha_condition_1}

		\item $ \alpha_i^k(y)=1 {\rm\ when\ } y=x_i^k$ for $k\in\{1,2,\dots,p\}$, 
		\hfill\refstepcounter{equation}\textup{(\theequation)}	
		\label{eq:alpha_condition_2}

		\item  $\alpha_i^1,\dots,\alpha_i^p$ are 
		continuous at all points in $\Omega_{C_i}$.
		\hfill\refstepcounter{equation}\textup{(\theequation)}	
		\label{eq:alpha_continuous}
	\end{itemize}
	The random variable $\gamma_i(y)$ then has the distribution: 
	\begin{equation}
	\label{eq:gamma}
	\begin{aligned}
	\gamma_i(y) = x_i^k {\rm\ with\ probability\ } \alpha_i^k(y) {\rm\ \ \ \ for\ } k \in \{1,2,\dots,p\}.
	\end{aligned}
	\end{equation}
	At a given point $y=y_i$ the values $\alpha^1_i(y_i), \alpha^2_i(y_i),\dots, \alpha^p_i(y_i)$
	serve as stochastic interpolation coefficients.
	\begin{figure}[hbt]
	\begin{center}
		\centering
		\includegraphics[width=0.45\textwidth]{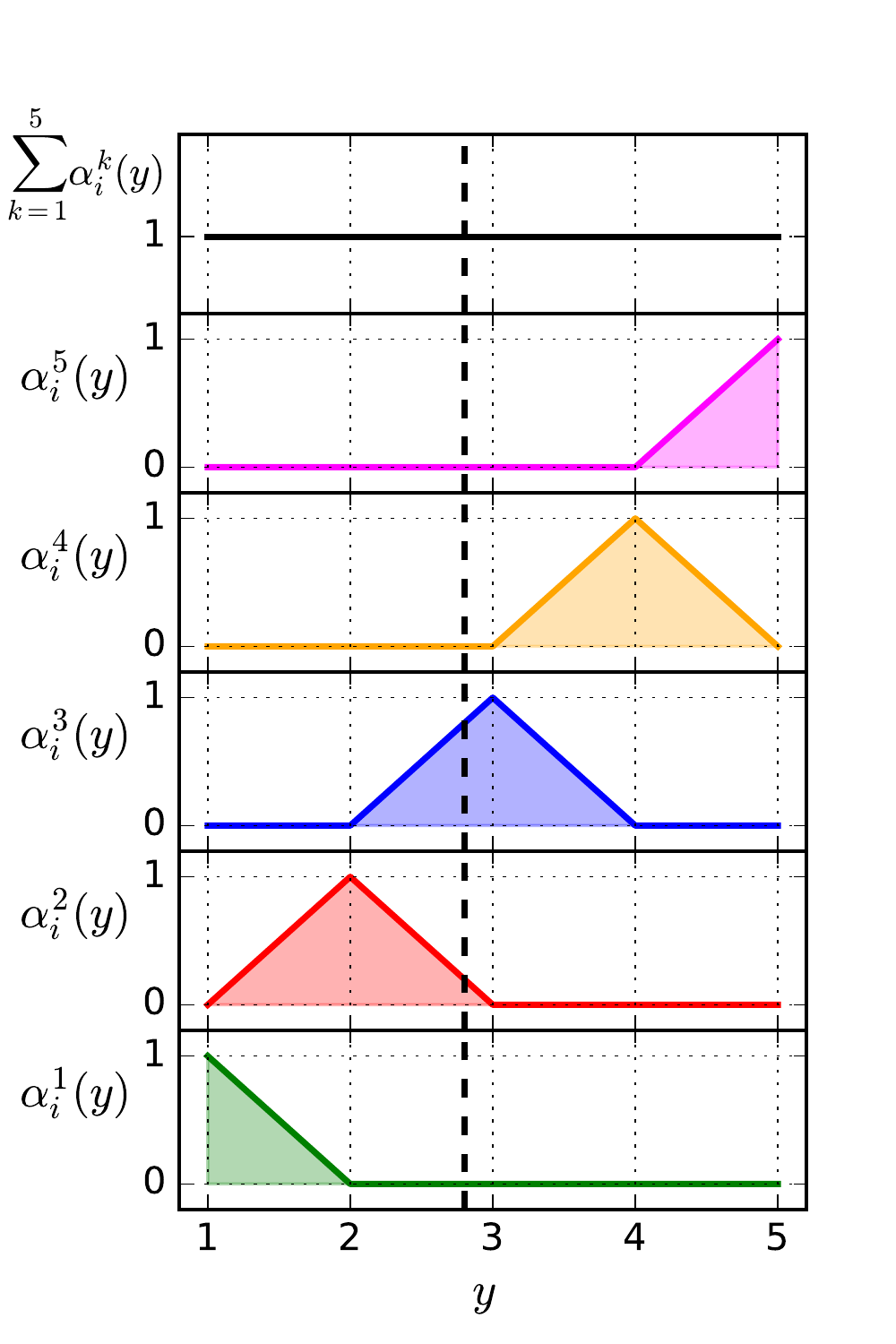}
		\vspace{-0.2em}
		\caption{An example for the set of stochastic interpolation coefficient functions 
		over the domain $\Omega_{D_i}=\{1,2,\dots,5\}$.
		At a given point $y=y_i$, $\alpha_i^k(y_i)$ represents the probability 
		with which the random variable $\gamma_i$ takes the value $x_i^k$.
		For example, at $y=2.8$ indicated by the dashed line,
		$\alpha_i^2 = 0.2$ and $\alpha_i^3 = 0.8$ whereas
		$\alpha_i^k = 0$ for $k \notin \{2,3\}$.} 
		\label{fig:Alphas_simple}
	\end{center}
	\end{figure}
	There are an infinite number of choices for the set of functions $\{\alpha_i^1, \alpha_i^2, \dots, \alpha_i^p\}$
	that satisfy conditions (\ref{eq:alpha_condition_1}) to (\ref{eq:alpha_continuous}).
	We illustrate one example of such a set in Figure \ref{fig:Alphas_simple}.
	In general, the shape of these functions will affect the resulting interpolation $\hat{f}$.
	In the following paragraphs we show that such a randomization can indeed produce
	continuous interpolations of the long-run average measure.
	The result is similar to \cite[Lemma 3]{T-ASE},
	except that the assumption about the ergodicity of the constituent Markov chains
	has been relaxed. This makes the analysis applicable to a 
	wider set of parameters and systems (including chains with 
	transient or periodic states such as those described in Section \ref{sec:ApplicationToQueues}).
	
	\subsection{Analysis}
	
	Let $X^1, X^2,\dots,X^m$ be points in the $n$-dimensional discrete parameter space $\Omega_D$.
	At each parameter point $X^j$ we assume that
	the behavior of the system can be described as a stationary Markov chain
	with a finite state-space $\mathcal{S}$ and a transition probability matrix $P^j$.
	Such a chain will have a unique stationary distribution
	(that is, a unique value of the distribution $\pi$ which satisfies $\pi P^j = \pi$)
	unless it contains two or more closed communicating classes.
	We assume that at each point $ X^j\in \Omega_D$ the corresponding chain  
	contains exactly one closed communicating class and therefore
	has a unique stationary distribution $\pi^j$.
	The chain is not required to be ergodic and 
	may contain periodic states and/or some transient states.
	Further, we assume that the chains at $X^1, \dots,X^m$ all share a common state-space $S$.
	This assumption holds when the model behavior is well-defined 
	for dynamically changing parameter-values,
	as demonstrated in Section \ref{sec:ApplicationToQueues}.
	Note that it is permissible for the subset of states forming a closed communicating class 
	at points $X^i$ and $X^j$ to be different or altogether non-overlapping for $i\neq j$.

	Let $c:\mathcal{S}\rightarrow\mathbb{R}$ be a cost function that assigns a fixed
	cost to every occurrence of a state in the Markov chain. Let $\pi^j_s$
	denote the probability of occurrence of a state $s\in\mathcal{S}$
	under the distribution $\pi^j$.
	The long-run average cost $f$ at the point $X^j$
	can then be defined in terms of the stationary distribution as follows:
	\begin{equation}\label{eq:long_run_avg_cost}
	f(X^j) = \sum_{s\in \mathcal{S}}{\pi^j_s c(s)}.
	\end{equation}
	Now consider the randomized model at $Y=(y_1, y_2, \dots, y_n) \in \Omega_C$ 
	where the $i^{th}$ parameter in the model 
	is assigned values of the random variable $\gamma_i(y_i)$.
	Here $\gamma_1,\gamma_2,\dots,\gamma_n$ are independent random variables
	with the distribution given by Equation (\ref{eq:gamma}).
	We assume that all parameters in the model are perturbed at 
	identical time instants.
	Let all of the stochastic interpolation coefficient functions 
	$\{\alpha_i^k:\Omega_{C_i}\rightarrow[0,1]\mid i\in\{1,2,\dots,n\}, k\in\{1,2,\dots,\abs{\Omega_{D_i}}\}\}$
	be of differentiability class $C^K$ (that is, have $K$ continuous derivatives, 
	where $K$ is some non-negative integer). 
	Let $\hat{f}:\Omega_C\rightarrow\mathbb{R}$
	be the corresponding long-run average measure of the randomized model.

	\begin{theorem}
	$\hat{f}$ is a $C^K$ interpolation of $f$.
	\end{theorem}
	
	\begin{proof}
	The $n$-dimensional vector of parameter values at 
	any instant of time is itself a random variable $\Gamma$
	whose distribution is a function of $Y$ as follows:
	\begin{equation}
	\label{eq:Gamma}
	\begin{aligned}
	\Gamma(Y) = X^j {\rm\ with\ probability\ } \beta^j(Y) {\rm\ \ \ \ for\ } j \in \{1,2,\dots,m\}.
	\end{aligned}
	\end{equation}
	Let $X^j = (x^j_1, x^j_2, \dots, x^j_n)$ be a point in $\Omega_D$ and
	let $I_i(x)$ denote the index of element $x$ in the set $\Omega_{D_i}$.
	The coefficients $\beta^j$ can then be obtained as:
	\begin{equation}\label{eq:beta}
	\begin{aligned}
	\beta^j(Y) &= \prod_{i=1}^{n}{\mathbb{P}(\gamma_i(y_i)=x^j_i)} &= \prod_{i=1}^{n}{\alpha_i^{I_i(x^j_i)}(y_i) }.
	\end{aligned}
	\end{equation}
	From equations (\ref{eq:beta}) and (\ref{eq:alpha_continuous})
	it follows that the functions $\beta^1,\beta^2,\dots,\beta^m$ which
	map values from the domain $\Omega_C$ to the range $[0,1]$ 
	also satisfy the following conditions:
	\begin{itemize}
	
		\item $\sum_{j=1}^{m}{\beta^j(Y)}=1 \ \ \ \  \forall Y\in\Omega_{C}$,
		\hfill\refstepcounter{equation}\textup{(\theequation)}
		\label{eq:beta_condition_1}

		\item $ \beta^j(Y)=1 {\rm\ when\ } Y=X^j$ for $j\in\{1,2,\dots,m\}$ ,
		\hfill\refstepcounter{equation}\textup{(\theequation)}	
		\label{eq:beta_condition_2}

		\item  $\beta^1, \dots, \beta^m$ are 
		continuous at all points in $\Omega_{C}$.
		\hfill\refstepcounter{equation}\textup{(\theequation)}\label{eq:beta_continuous} 
		\label{eq:beta_condition_3}
	\end{itemize}
	Let $P(Y)$ denote the transition probability matrix of the
	randomized model. We choose the time instants at which to perturb the 
	parameter values in such a way that $P(Y)$ is given by
	\begin{equation}\label{eq:PY}
	P(Y) = \sum_{j=1}^{m}{\beta^j(Y) P^j}.
	\end{equation}
	In a discrete-time system, this can be achieved in a straightforward
	manner by perturbing the parameter values at the beginning of each time-slot.
	From (\ref{eq:PY}) and (\ref{eq:beta_continuous}) it follows 
	that $P(Y)$ is continuous with respect to $Y$.
	We now refer to a result from \cite[Section~6]{Schweitzer}
	which states that:{ \em If $P^A$ is the transition probability matrix
	of a finite Markov chain containing a single irreducible subset of states
	(a single closed communicating class), then for an arbitrary
	stochastic matrix $P^B$ with the same state-space as $P^A$, 
	the randomized stationary Markov chain with transition probability
	matrix 
	$$P(\lambda)=(1-\lambda)P^A+\lambda P^B\ \ \ \ 0\leq\lambda<1$$
	will also possess a single irreducible subset of states. Further, $P(\lambda)$ has a 
	unique stationary distribution $\pi(\lambda)$ which is infinitely differentiable
	with respect to $\lambda$ for $\lambda\in[0,1)$.}

	In Equation (\ref{eq:PY}) $P^1, P^2, \dots, P^m$ each have a single
	irreducible set of states. Therefore the stationary distribution
	$\pi(Y)$ corresponding to $P(Y)$ exists and is infinitely differentiable
	with respect to the coefficients $\beta^1(Y),\dots,\beta^m(Y)$
	and $K$-times continuously differentiable ($C^K$) with respect to $Y$.
	Let $\pi_s(Y)$ denote the probability of occurrence of a state $s$
	under the distribution $\pi(Y)$.
	The long-run average cost $\hat{f}$ in the randomized model is given by:
	$$\hat{f}(Y) = \sum_{s\in \mathcal{S}}{\pi_s(Y)c(s)}$$
	The function $\hat{f}$ is also $C^K$ with respect to $Y$.
	From equations (\ref{eq:Gamma}) and (\ref{eq:beta_condition_2}) we have $\Gamma(Y) = X^j$ 
	with probability $1$ when $Y = X^j$ for $j \in \{1,2,\dots,m\}$.
	Thus $\hat{f}(Y) = f(Y)$ whenever $Y\in\Omega_D$.
	Therefore $\hat{f}$ is a $C^K$ interpolation of $f$.
	\end{proof}

	Thus we have shown that a randomization of the simulation model 
	can be used as a means of producing continuous interpolations of the
	long-run average measure $f$ under the listed assumptions.
	Note that for $\hat{f}$ to be of class $C^K$, it is sufficient but not
	necessary for the stochastic interpolation coefficients $\alpha_i^1,\alpha_i^2,\dots$
	to be $C^K$ functions. For instance, it may be possible to obtain 
	a continuously differentiable interpolation $\hat{f}$
	using coefficients that are not differentiable at integer points.
	%

	\section{Generating the Stochastic Interpolation Coefficients}\label{sec:GeneratingAlphas}

	Consider the set of functions
	$\{\alpha_i^k:\Omega_{C_i}\rightarrow[0,1]\mid k\in\{1,2,\dots,p\}\}$
	that satisfy conditions (\ref{eq:alpha_condition_1}) to (\ref{eq:alpha_continuous}).
	While an infinite number of choices exist for such a set of functions, we 
	present one possible template that can be used to generate a family of such sets.
	The template is useful as the final interpolation curves
	can be tuned by varying the parameters of this template.

	For a given point $y\in\Omega_{C_i}$ we refer to the {\em stencil} around point $y$,
	denoted $\mathbb{S}(y)$ as the set of points in $\Omega_{D_i}$
	that are chosen as basis to represent $y$.
	In other words $\mathbb{S}(y)$ is the set of values that the random
	variable $\gamma_i(y)$ would take with non-zero probabilities.
	For example, at $y=6.8$ 
	a symmetric stencil of size $2$ would be the set $\{6, 7\}$.
	As a concrete example for the case where $\Omega_{D_i}$ consists of
	consecutive integers, a symmetric stencil 
	of size $\le2N$ (where $N$ is a natural number) 
	around any point $y\in\Omega_{C_i}$ 
	can be defined as follows:
	\begin{equation}\label{eq:stencil}
	\begin{aligned}
	\mathbb{S}(y)=
	\begin{cases}
	\{y\} & \text{if } y \in \Omega_{D_i}\\
	\{ (\floor{y}-N+1),\dots,\floor{y},\ceil{y},\dots,(\ceil{y}+N-1)\}\cap\Omega_{D_i} & \text{otherwise.}
	\end{cases}
	\end{aligned}
	\end{equation}
	Note that the stencil size can be lower than $2N$ 
	for points near the boundary of $\Omega_{D_i}$.
	Now corresponding to each point $x^k \in\Omega_{D_i}$
	we define a function $L^k(y)$ such that $L^k(y)$ is non-negative,
	continuous over $\Omega_{C_i}$ and its value approaches $0$ as $y$ 
	approaches points in its stencil $\mathbb{S}(y)$ other than $x^k$.
	For example:
	\begin{equation}
	\begin{aligned}
	L^k(y) = \prod_{\substack{j \in\mathbb{S}(y)\\  j \neq x^k}}{\abs{y - j}}.
	\end{aligned}
	\end{equation}
	The stochastic interpolation coefficients $\alpha_i^k(y)$ for $k\in\{1,2,\dots,p\}$ can now be defined as:
	\begin{equation}\label{eq:alpha_template}
	\begin{aligned}
	\alpha_i^k(y)=
	\begin{cases}
	0,					              		& \text{if }  x^k\notin\mathbb{S}(y)\\
	\dfrac{L^k(y)}{\sum_{j\in\mathbb{S}(y)}{L^j(y)}}              & \text{if } x^k\in\mathbb{S}(y).
	\end{cases}
	\end{aligned}
	\end{equation}
	Since the values of $L^k(y)$ are non-negative and a normalization 
	is performed in Equation (\ref{eq:alpha_template}) it can be seen
	that the coefficients $\alpha_i^k(y)$ will always lie in the range $[0,1]$
	and satisfy conditions (\ref{eq:alpha_condition_1}) to (\ref{eq:alpha_continuous}).
	To obtain a more general template, let $s$ and $r$ be fixed constants such that 
	$s\in \mathbb{R}{\setminus\{0\}}$ and $r\in\mathbb{R}_{>0}$.
	Let $m$ be the smallest element in $\mathbb{S}(y)$. Then,
	\begin{equation}\label{eq:L_template}
	\begin{aligned}
	L^k(y) = \prod_{\substack{j\in\mathbb{S}(y)\\  j \neq x^k}}{ {\abs{(y-m+1)^s-(j-m+1)^s}}^r}.
	\end{aligned}
	\end{equation}
	The coefficients $\alpha_i^k(y)$ can be obtained as given by Equation (\ref{eq:alpha_template})
	using the values of $L^k(y)$ computed using Equation (\ref{eq:L_template}).
	\begin{figure}[ph]
	\begin{center}
		\begin{subfigure}{0.45\textwidth}
		\centering
		\caption{stencil size $=2$, $s=1,r=1$}
		\vspace{-0.3em}
		\includegraphics[width=0.8\textwidth]{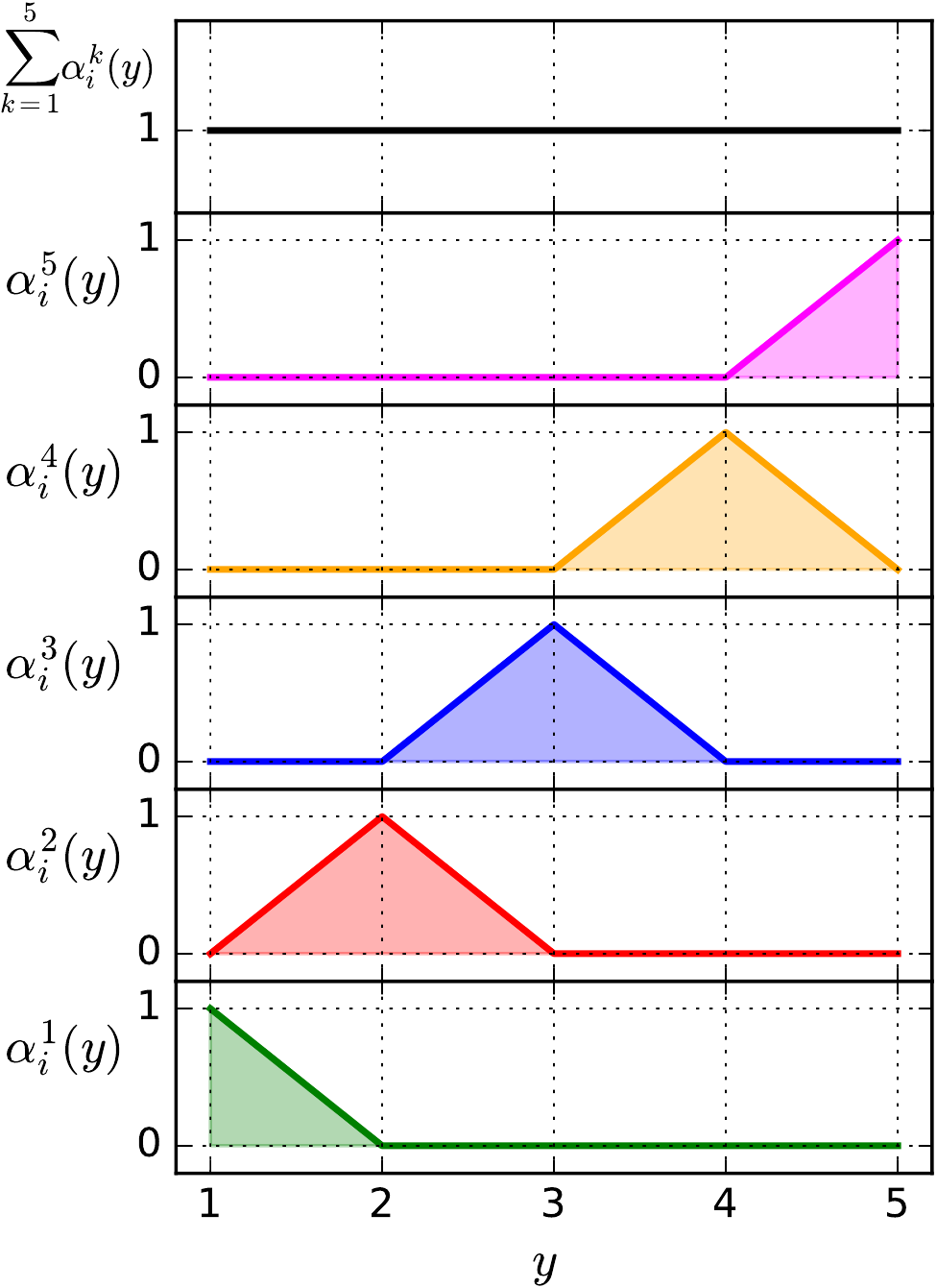}
		\label{fig:alphas_1}
		\end{subfigure}\hspace{2em}
		\begin{subfigure}{0.45\textwidth}
		\centering
		\caption{stencil size $=4$, $s=1,r=1$}
		\vspace{-0.3em}
		\includegraphics[width=0.8\textwidth]{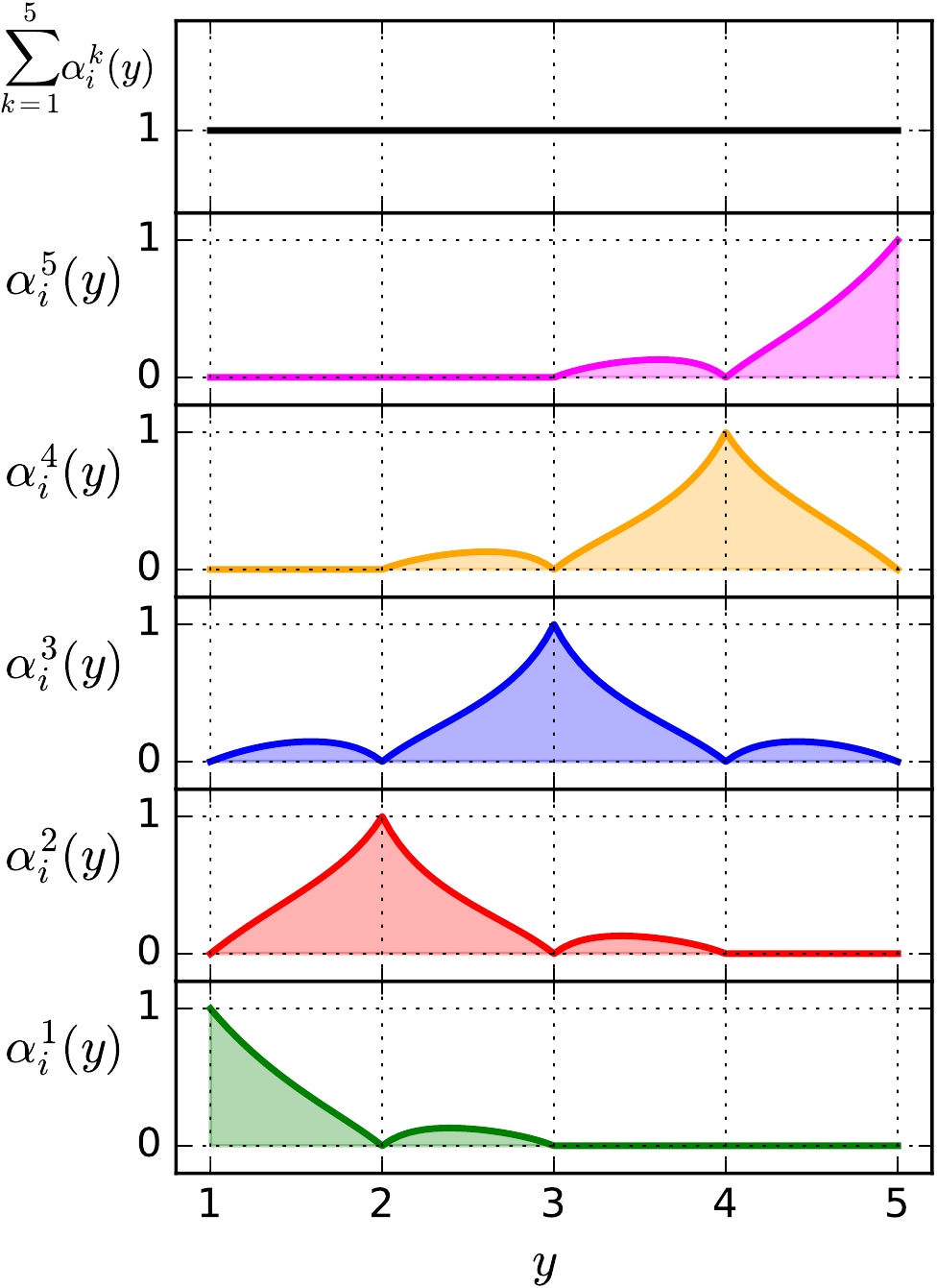}
		\label{fig:alphas_2}
		\end{subfigure}
		
		\vspace{0.2em}
		
		\begin{subfigure}{0.45\textwidth}
		\centering
		\caption{stencil size $=2$, $s=1$, $r=2$}
		\vspace{-0.3em}
		\includegraphics[width=0.8\textwidth]{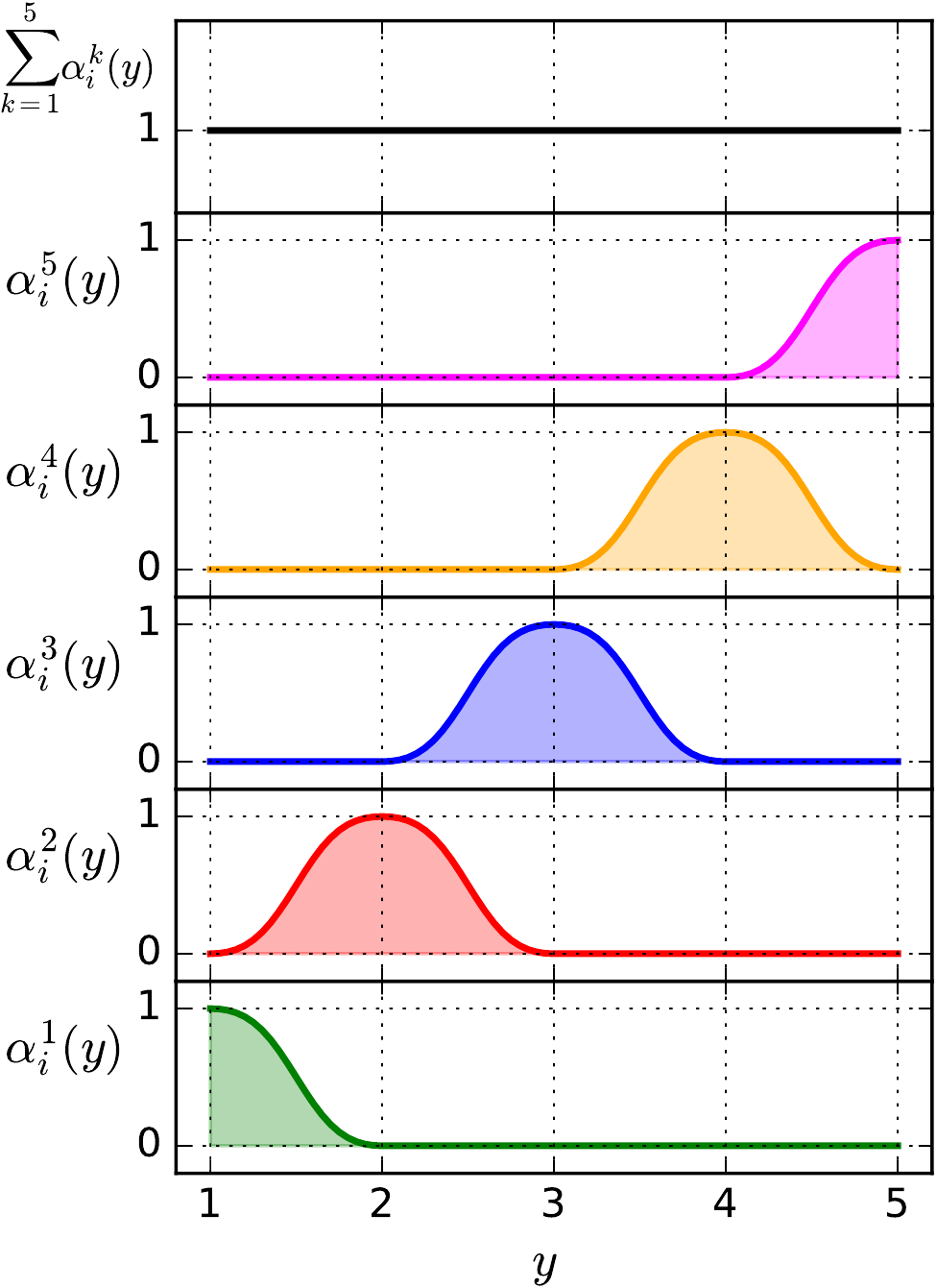}
		\label{fig:alphas_3}
		\end{subfigure} \hspace{2em}
		\begin{subfigure}{0.45\textwidth}
		\centering
		\caption{stencil size $=2$, $s=1$, $r=0.5$}
		\vspace{-0.3em}
		\includegraphics[width=0.8\textwidth]{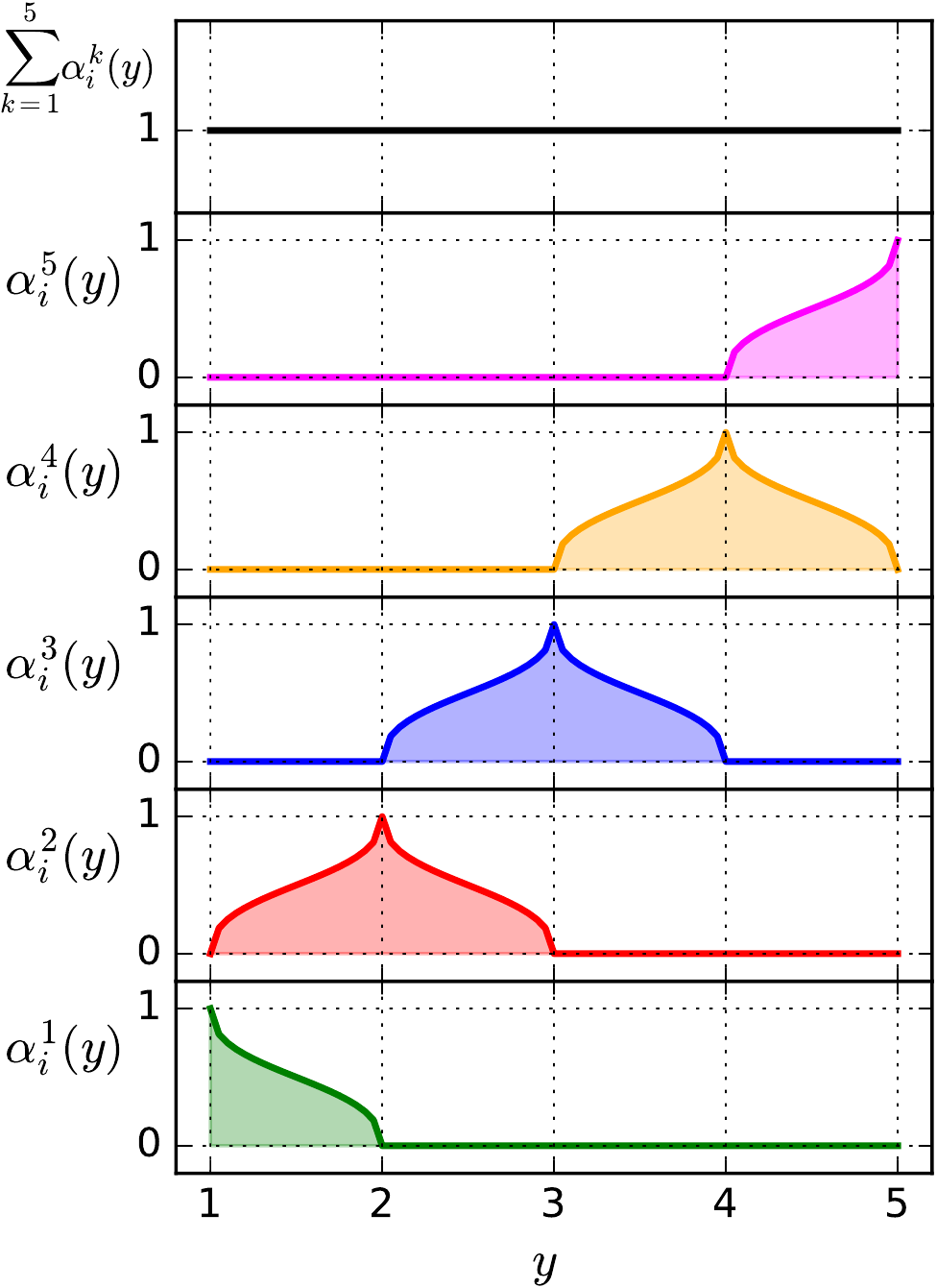}
		\label{fig:alphas_4}
		\end{subfigure}
		
		\vspace{0.2em}
		
		\begin{subfigure}{0.45\textwidth}
		\centering
		\caption{stencil size $=2$, $s=4$, $r=1$}
		\vspace{-0.3em}
		\includegraphics[width=0.8\textwidth]{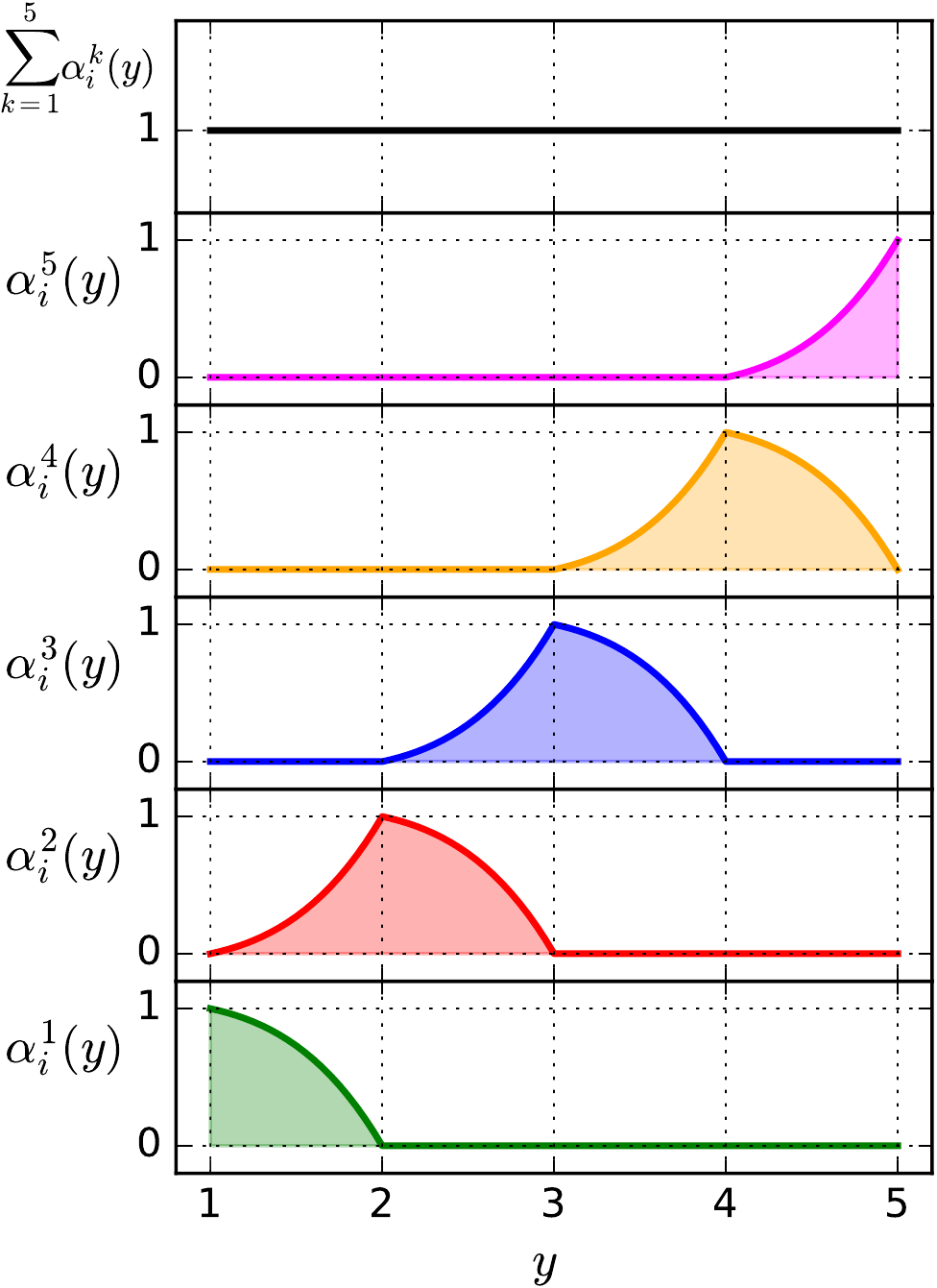}
		\label{fig:alphas_5}
		\end{subfigure} \hspace{2em}
		\begin{subfigure}{0.45\textwidth}
		\centering
		\caption{stencil size $=2$, $s=-2$, $r=1$}
		\vspace{-0.3em}
		\includegraphics[width=0.8\textwidth]{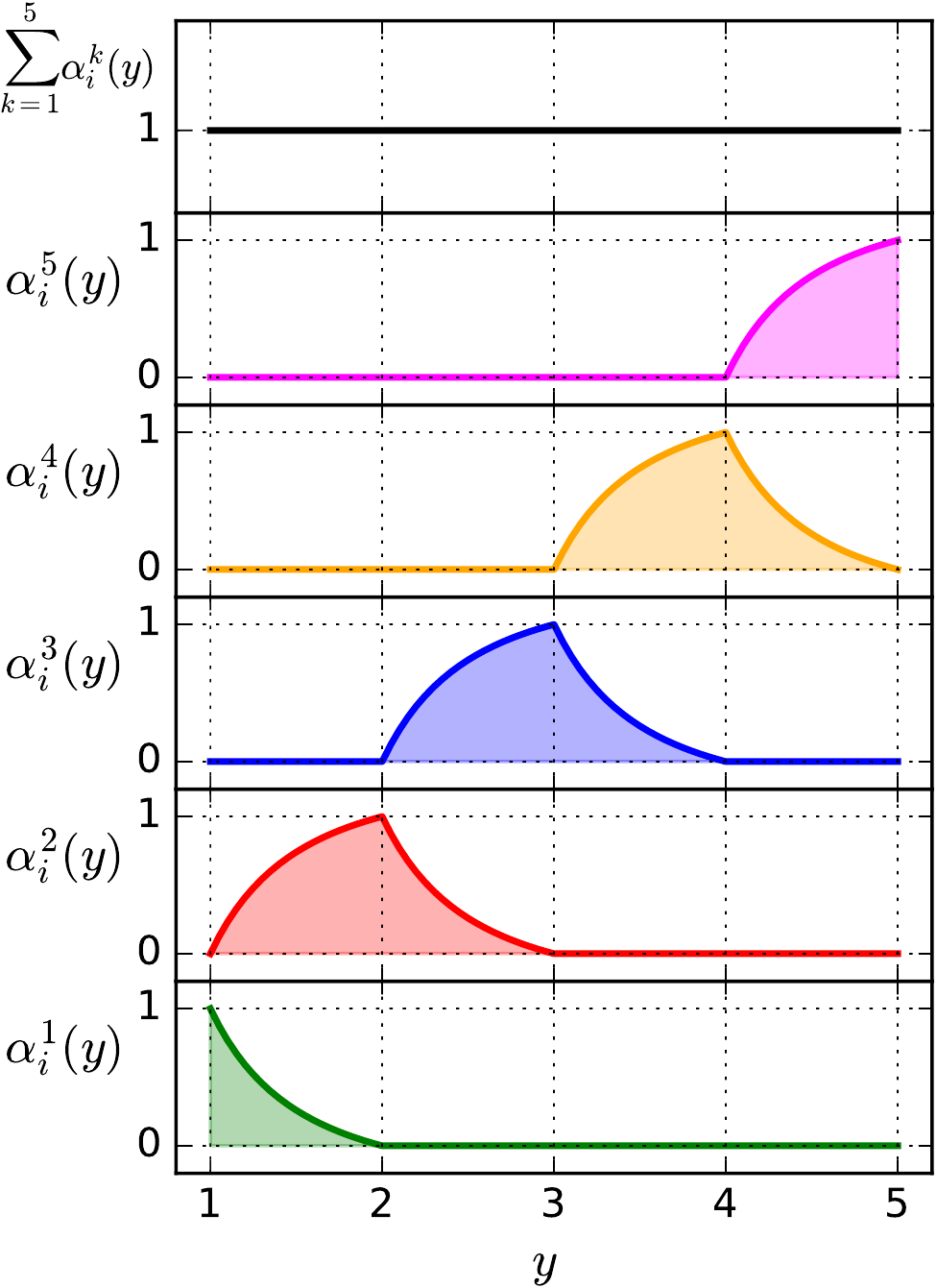}
		\label{fig:alphas_6}
		\end{subfigure}
		\vspace{-1em}
		\caption{Stochastic interpolation coefficients over the domain $\Omega_{C_i}=[1,5]$} 
		\label{fig:Alphas}
	\end{center}
	\end{figure}
	
	The stencil size ($2N$), and the constants $s$ and $r$ are parameters
	of this template.
	Figure \ref{fig:Alphas} shows sets of coefficients $\alpha_i^1,\dots,\alpha_i^p$ generated
	using this template for different values of $s$, $r$ and the stencil size.
	The constant $r$ controls the spread of $\alpha_i^k(y)$
	around the point $y=x^k$. The spread is linear for 
	$r=1$, and reduces with increasing values of $r$. It is 
	interesting to note that for very large values of $r$ 
	(when the stencil size is $2$ and $s=1$)
	the value of $\alpha_i^k(y)$ is nearly $1$ in the range $(x^k-0.5,x^k+0.5)$
	and the interpolation approaches a simple nearest-neighbour rounding.
	The constant $s$ controls the skew. The functions are symmetric for $s=1$,
	show a right-skew for $s>1$ and a left skew for $s<1$.
	The template thus provides a simple means of varying the 
	shape of the stochastic interpolation coefficients.
	In the next section, we demonstrate how this 
	can be used to tune the final interpolation curves
	for a specific example.

\section{Application to Discrete-time Queues}\label{sec:ApplicationToQueues}
	
	We now present concrete examples for the application of the embedding technique
	to discrete-time queues.
	We use the following notations and assumptions throughout this section:
	Time is divided into unit-sized slots.  
	Jobs arrive into the system near the beginning of a slot and 
	depart towards the end of the slot. At most one job can arrive within a single slot.
	The queue state is measured at the end of a slot, as illustrated in Figure \ref{fig:DiscreteTime}. 
	$S_0$ denotes the initial state and $S_t$ denotes the state measured at the end of slot $t$.
	For job arrivals with geometrically distributed inter-arrival times (denoted Geo) 
	the probability of a job arriving in a slot is denoted as $p$. 
	For geometrically distributed (Geo) service times, the probability of the server
	finishing a job in a slot is denoted as $q$. This probability is independent of 
	the number of slots for which a job has already received service.
	For a deterministic server (denoted D), the number of slots taken
	to process a job is $T$.
	\begin{figure}[h]
	\begin{center}
	\includegraphics[width=0.8\textwidth]{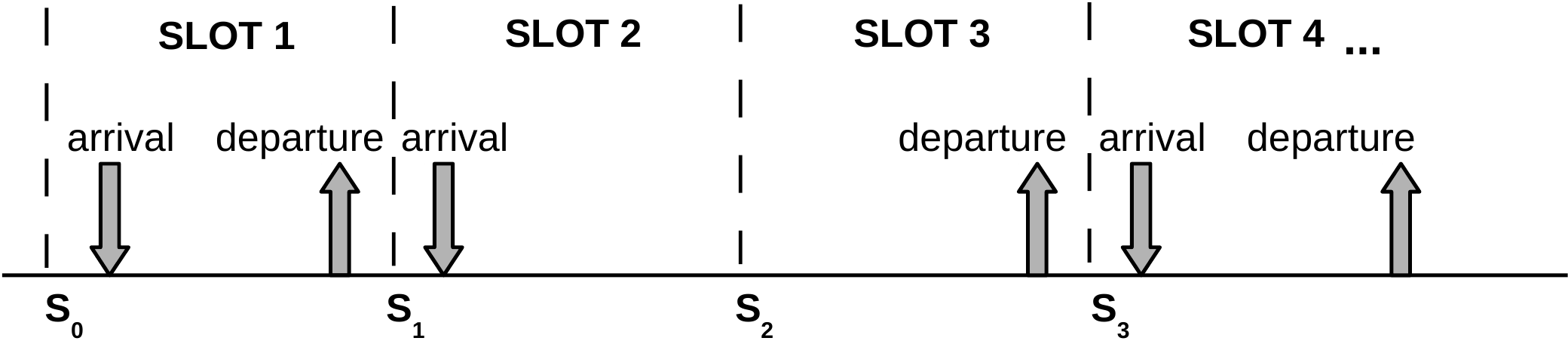}
	\caption{A Discrete-time Queue}
	\label{fig:DiscreteTime}
	\end{center}
	\end{figure}

	\subsection{Embedding Queue-Capacity}
	Consider a Geo/Geo/1 queue with finite buffering. 
	The queue state $S_t$ is the number of jobs in the queue at the end of slot $t$. 
	The queue has a capacity parameter $C\in\mathbb{N}$ that we wish to embed
	into a continuous space. We first define the behavior of the queue with 
	respect to $C$ as follows:
	\begin{definition}
	A job arriving in slot $t$ is allowed to enter the queue 
	if $S_{t-1} < C$, else the job is lost.
	\end{definition}
	
	By defining the capacity parameter in this way, we ensure that
	the Markov chains corresponding to every possible value of $C$ share the
	same state-space. The states $\{S_t\mid S_t > C\}$ are transient
	and unreachable from states $\{S_t\mid S_t\leq C\}$. As an example,
	Figure \ref{fig:MarkovChains_C} shows 
	the Markov chains for $C=1$ and $C=2$.
	\begin{figure}[h]
	\begin{center}
	\includegraphics[width=0.8\textwidth]{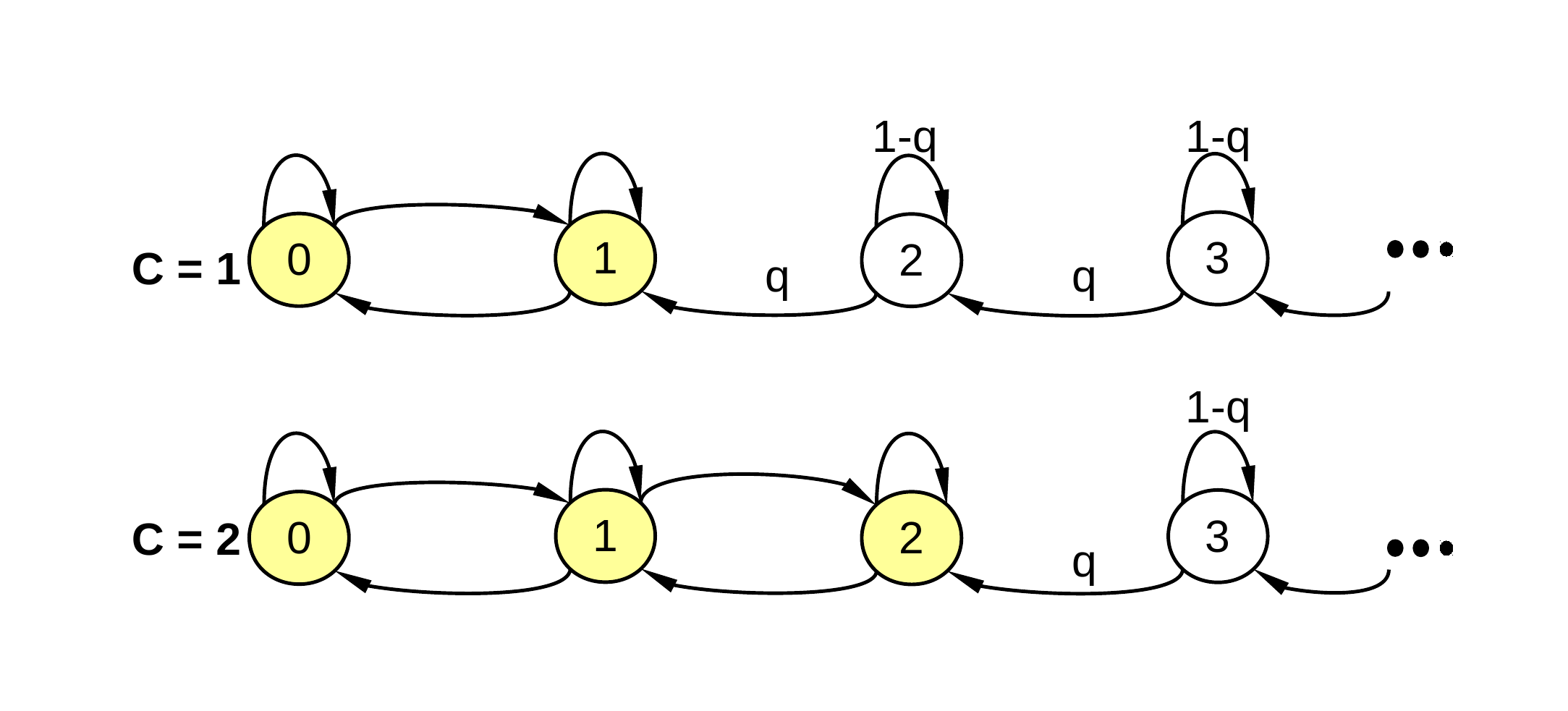}
	\caption{Markov chains for the finite capacity Geo/Geo/1 queue. (The non-shaded states
	are unreachable from the shaded states)}
	\label{fig:MarkovChains_C}
	\end{center}
	\end{figure}
	Let $f:\mathbb{N}\rightarrow\mathbb{R}$ be some long-run average measure of this system
	expressed as a function of the capacity parameter $C$.
	We wish to embed the capacity parameter into continuous space, and evaluate the
	interpolation $\hat{f}(y)$ of $f$ at some given point $y\in\mathbb{R}_{\ge 1}$.
	To do so, we construct a randomized model where the capacity parameter
	is perturbed at the beginning of each slot and assigned the value
	of the random variable $\gamma_c(y)$ with the distribution:
	$$\gamma_c(y)= k {\rm \ with\ probability\ } \alpha^k(y)\qquad{\rm for\ } k\in\{1,2,3,\dots\},$$
	where the coefficients $\alpha^k$ satisfy the conditions described in Section \ref{sec:AGeneralRandomizaionScheme}.
	Let $C_t$ denote the instantaneous value
	of the capacity parameter for the duration of slot $t$. 
	By the definition of the capacity parameter above, whenever the 
	parameter is updated the jobs already present in the queue
	are not disturbed. The updated value of the parameter
	is used solely to decide if a new job should
	be accepted into the queue. 
	Thus it is possible that $S_t > C_t$ for some values of $t$.

	In Figure \ref{fig:GG1_C}, we show the simulation results obtained with this
	randomization scheme. The interpolation coefficients $\alpha^k$
	are generated using the template described in Section \ref{sec:GeneratingAlphas}
	with the spread and skew attributes set to 
	$r=1$ and $s=-1$. For all simulation results presented in this
	section we have used a 2-point stencil.
	We fix the arrival and service probabilities $p,q$ and sweep the
	queue capacity parameter $y$ in steps of $0.05$.
	The interpolated function $\hat{f}(y)$ is the blocking probability
	(probability of an arriving job being denied entry into the system).
	Each point in the plot is the mean value of $\hat{f}(y)$ measured using $100$
	simulation samples with distinct randomization seeds. 
	The shaded area around the plot represents the $\pm 3$ standard-deviation interval.
	\begin{figure}[h]
		\centering
		\includegraphics[width=0.8\textwidth]{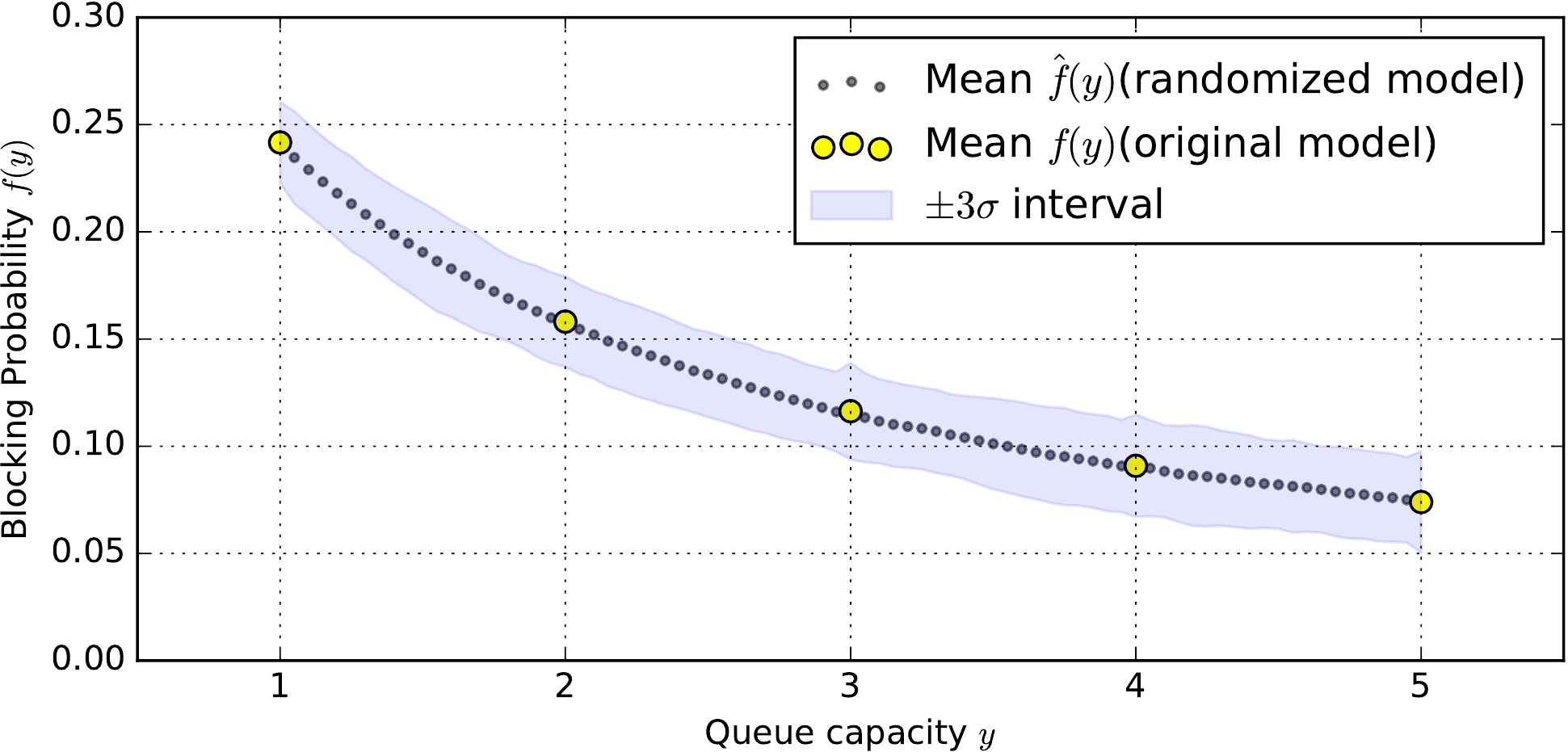}
		\caption{Blocking probability $\hat{f}(y)$ versus the queue capacity $y$
		in a finite capacity Geo/Geo/1 queue ($p=0.5$, $q=0.51$, simulation length$=10^4$ slots)}
		\label{fig:GG1_C}
	\end{figure}

	We observe that the randomization scheme produces a smooth interpolation
	and the standard deviation values (indicating simulation error)
	are similar at the discrete and the interpolated points. 
	The computational overhead of the embedding
	contributed primarily by the additional calls to a random number generator,
	was between $5\%$ to $10\%$. 
	The overhead was computed as the relative difference between 
	the time per simulation for the original discrete-parameter model (at some $y=y_0\in\mathbb{N}$)
	and the randomized model (at $y_0+0.5$). Thus a smooth embedding of the
	queue capacity parameter could be obtained through a randomization of the
	simulation model.

	\subsection{Tuning the Interpolation}
	The interpolation is sensitive to the shape of the 
	coefficient function $\alpha^k(y)$ used during randomization.
	To demonstrate this, we consider the problem of embedding
	the queue capacity parameter in a finite capacity Geo/D/1 queue. This queue
	is similar to the Geo/Geo/1 queue except that the server is deterministic
	with a fixed service time of $T$ slots.
	The queue capacity parameter is embedded using a randomization scheme
	identical to that of the Geo/Geo/1 case using the template described
	in Section \ref{sec:GeneratingAlphas}. 
	We show the simulation results and the resulting interpolation
	obtained using different values of the skew ($s$) and spread ($r$) parameters
	in Figure \ref{fig:GD1_C}. 
	\begin{figure}[h!]
		\centering
		\includegraphics[width=1\textwidth]{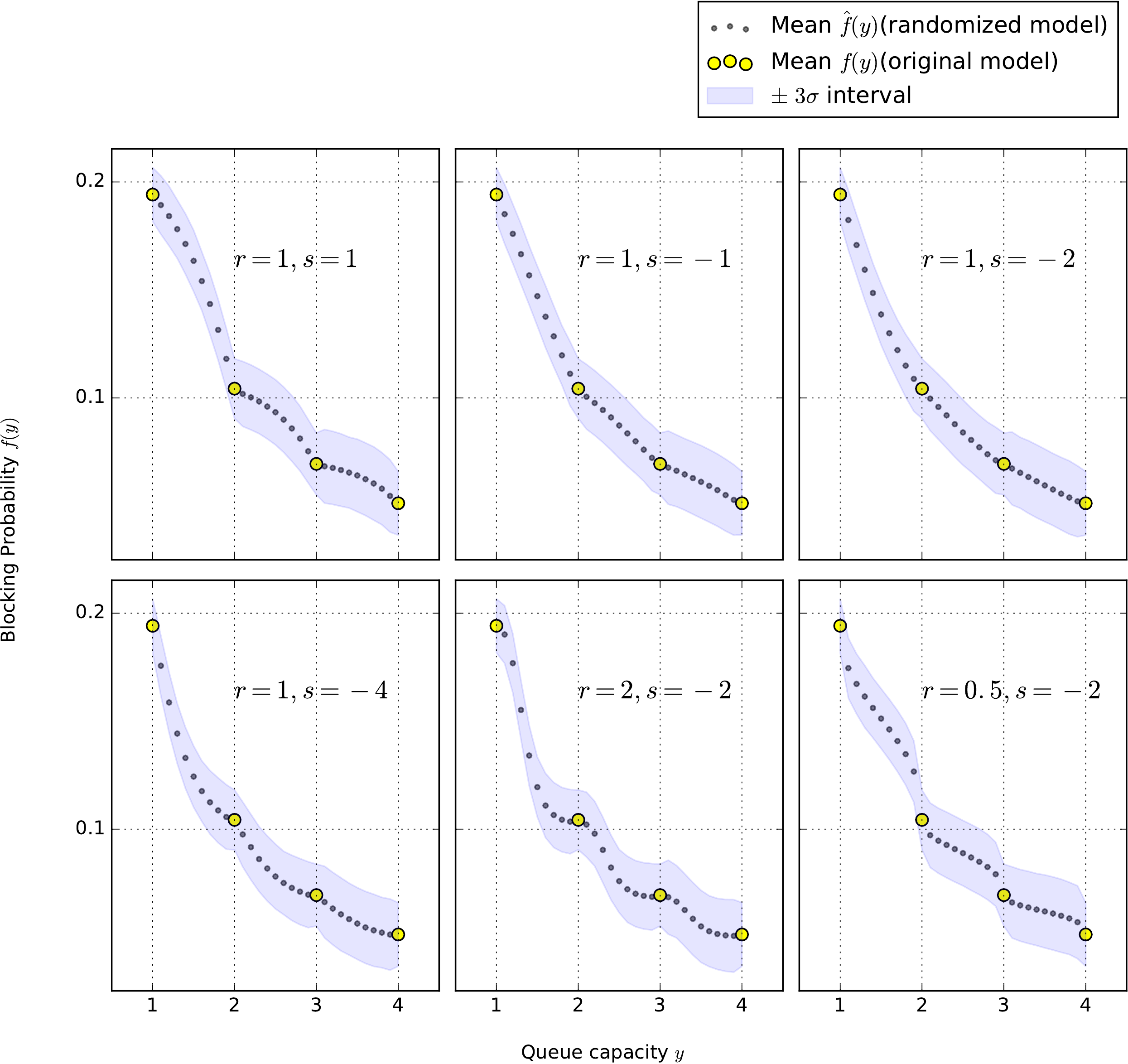}
		\caption{Interpolation results obtained for different values of the 
		template parameters $s$ and $r$. The objective function $\hat{f}(y)$
		is the blocking probability for the queue capacity $y$
		in a finite capacity Geo/D/1 queue ($p=0.49$, $T=2$, simulation length$=10^4$ slots)}
		\label{fig:GD1_C}
	\end{figure}

	We observe that for $r=1,s=1$ the interpolation has kinks at the integer points.
	If the original objective function $f$ is integer-convex
	over some convex subset $\Omega_{D_s}$ of the domain $\Omega_D$, it is desirable that the interpolation 
	$\hat{f}$ also be convex over the region that forms the convex-hull of $\Omega_{D_s}$.
	Otherwise a continuous optimizer applied directly over $\hat{f}$
	can get stuck in a local minimum that does not correspond to any 
	minimizer of $f$ in $\Omega_{D_s}$. Further, it is also
	desirable that the interpolation $\hat{f}$ be smooth
	so that gradient-based continuous optimizers can converge fast.
	Considering these criteria, a reasonably good interpolation was obtained
	at $r=1,s=-2$.
	For some systems, it may be possible to arrive at the best values
	for the randomization parameters ($s,r$) analytically. This is
	an interesting problem, however beyond the scope of the current work.
	For examples presented henceforth, we have tuned the interpolation 
	curves along each parameter axis by varying the template parameters.

\subsection{Embedding the Number of Servers}

	Consider a Geo/Geo/K queue. The queue has infinite buffering
	and $K$ identical, independent servers working in parallel.
	We wish to embed the parameter $K\in \mathbb{N}$ into
	continuous space. To do so, we first re-define the
	system behavior as follows:
	
	\begin{definition}
	The system consists of a single, infinitely fast 
	{\em controller} with a parameter $K$, and a fixed large number ($>K$) 
	of identical, independent servers. In each time-slot, the controller pulls 
	jobs from the head of the queue and assigns each job to a free server, as long as 
	the queue is not empty and the number of jobs currently receiving service is less than $K$.
	\end{definition}

	Note that for a fixed (integer) value of $K$, this description is identical to a queue 
	with $K$ independent servers. However, the new definition of the queue behavior
	makes it intuitive for the controller parameter $K$ to be updated dynamically.
	To embed $K$ into continuous space and evaluate the interpolation
	at some point $y\in\mathbb{R}_{\ge 1}$, we construct a
	randomized model where $K$ is perturbed at the beginning of 
	each slot and assigned the value of the random variable $\gamma_k (y)$.
	By the definition above, whenever the parameter $K$ is
	updated, the jobs already receiving service are not disturbed and the
	updated parameter value is used solely for deciding if service can
	commence for new jobs.

	In Figure \ref{fig:GGK_K}, we show the interpolation obtained 
	using this scheme with $s=1$ and $r=1$.
	We fix the arrival and service probabilities $p,q$ and sweep the
	parameter $y$ in small steps (the step-size is chosen to be
	smaller near the knee region). 
	The interpolated function $\hat{f}(y)$ is the average number
	of jobs in the system. Each point in the plot is the mean value of 
	$\hat{f}(y)$ measured using $100$ simulation samples.
	\begin{figure}[t]
		\centering
		\includegraphics[width=0.75\textwidth]{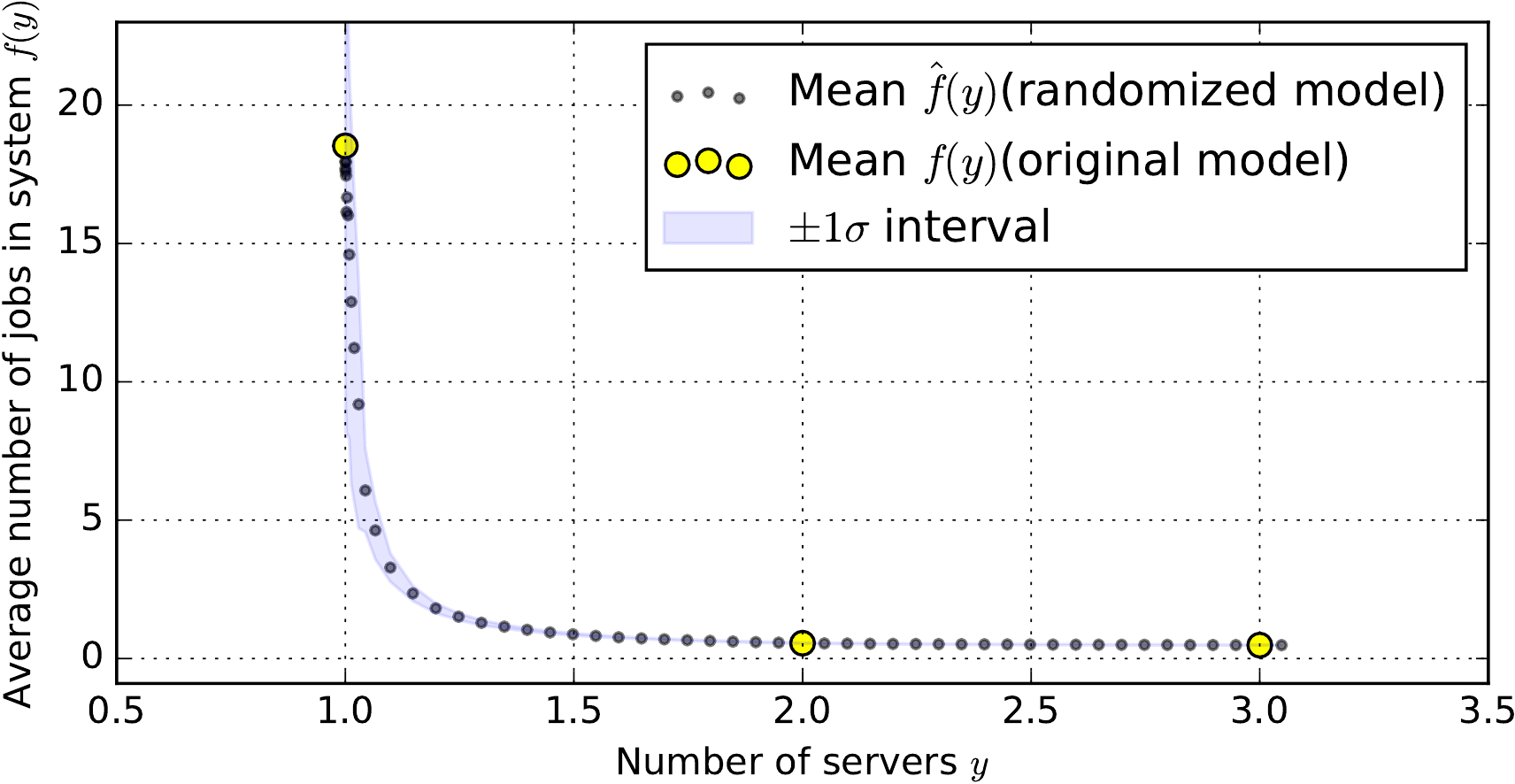}
		\caption{Average number of jobs in the system $\hat{f}(y)$ versus the number of servers ($K=y$)
		in a Geo/Geo/K queue ($p=0.5$, $q=0.51$, simulation length$=10^4$ slots)}
		\label{fig:GGK_K}
	\end{figure}
	\begin{figure}[h]
		\centering
		\includegraphics[width=0.75\textwidth]{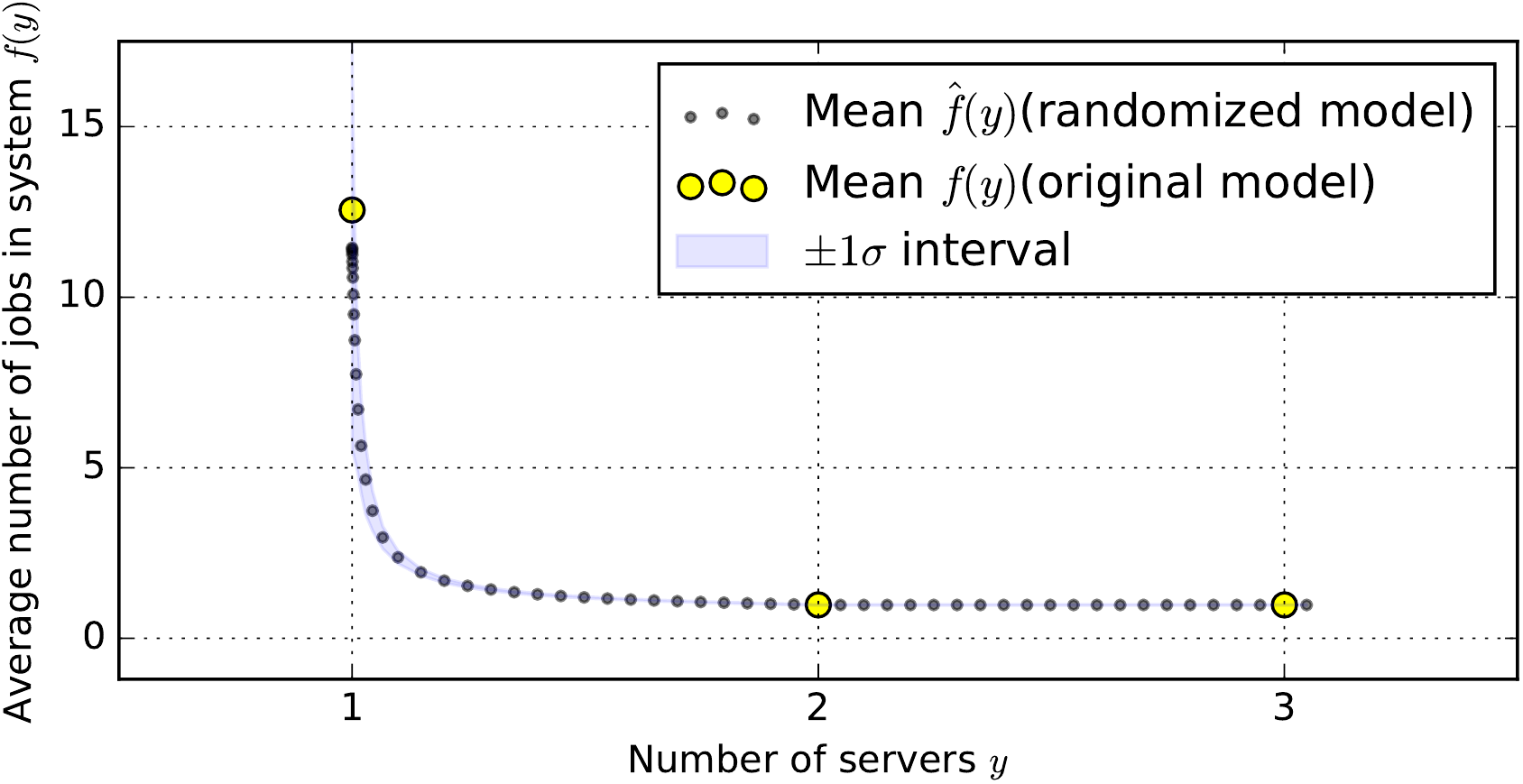}
		\caption{Average number of jobs in the system $\hat{f}(y)$ versus the number of servers ($K=y$)
		in a Geo/D/K queue ($p=0.49$, $T=2$, simulation length$=10^4$ slots)}
		\label{fig:GDK_K}
	\end{figure}
	In Figure \ref{fig:GDK_K} we show the interpolation results
	for the same randomization scheme for a Geo/D/K queue 
	(with a deterministic service time of $T$ slots).

\subsection{Embedding Service Time}

	Consider a Geo/D/1 queue. The server is deterministic
	with a fixed service time of $T\in\mathbb{N}$ slots.
	To embed the parameter $T$ into continuous space, 
	we first define the server behavior as follows:
	
	\begin{definition}
	The server has a parameter $T$. At the end of each slot,
	the server ends jobs that have already received $\ge T$ slots of service.
	\end{definition}
	
	To embed $T$ into continuous space and evaluate the interpolation
	at some point $y\in\mathbb{R}_{\ge 1}$, we construct a
	randomized model where $T$ is perturbed at the beginning of 
	each slot and assigned the value of the random variable $\gamma_T(y)$.
	In Figure \ref{fig:GD1_T}, we show the interpolation obtained 
	using this randomization scheme with $s=1$ and $r=1$.
	We fix the arrival probability $p$ and sweep the service time
	parameter $y$ in steps of 0.05.
	The interpolated function $\hat{f}(y)$ is the average number
	of jobs in the system. Each point in the plot is the mean value of 
	$\hat{f}(y)$ measured using $100$ simulation samples. The randomization
	produces a smooth interpolation of $f$.
	\begin{figure}[t]
		\centering
		\includegraphics[width=0.75\textwidth]{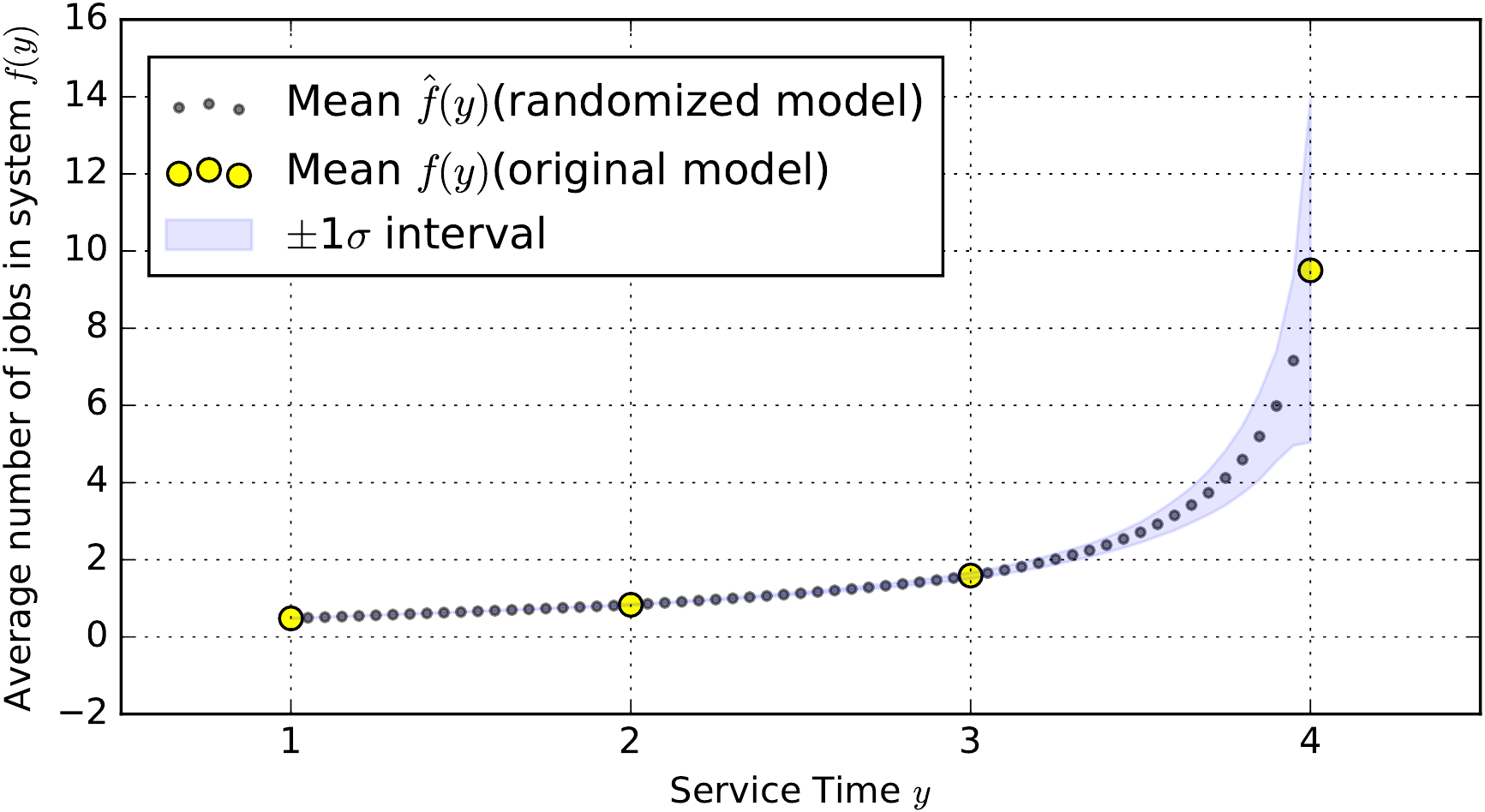}
		\caption{Average number of jobs in the system $\hat{f}(y)$ versus the service time ($y$)
		in a Geo/D/1 queue ($p=0.24$, simulation length$=10^4$ slots)}
		\label{fig:GD1_T}
	\end{figure}

	Using the approach described in this section, 
	multiple discrete parameters in a model can be embedded
	simultaneously and independently of each other, and existing 
	continuous optimizers can be effectively applied over such an embedding.
	We demonstrate this using an optimization case study in the following section.

\section{An Optimization Example}\label{sec:OptimizationCaseStudy}
	
	\subsection{Problem Statement}
	To demonstrate the utility of the embedding technique,
	we consider the optimization of a queueing network
	in Figure \ref{fig:OptimizationProblem}.
	\begin{figure}[tbh]
		\centering
		\includegraphics[width=0.8\textwidth]{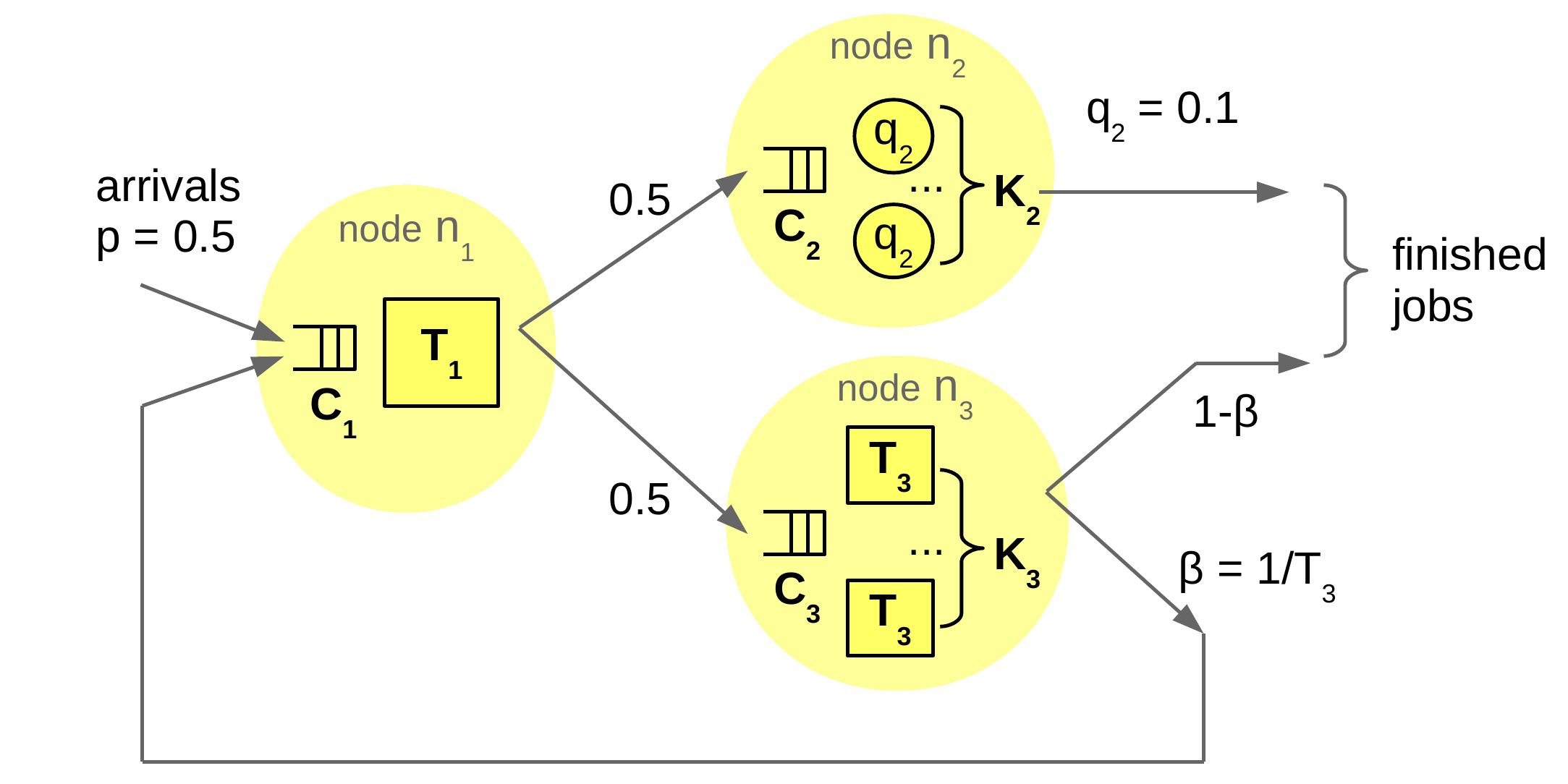}
		\caption{Queueing network to be optimized}
		\label{fig:OptimizationProblem}
	\end{figure}
	\begin{denseItemize}
	
	\item The system consists of three nodes $n_1$, $n_2$ and $n_3$.
	Each node $n_i$ has a queue with a finite capacity $C_i$.
	
	\item Jobs arrive at $n_1$ with geometrically distributed inter-arrival times
	(with an arrival probability $p$). An arriving job that finds the queue full is lost.
	The node $n_1$ consists of a single deterministic server with a service time of $T_1$ slots.
	This server forwards each job to either $n_2$ or $n_3$ with equal probabilities,
	and stalls if the destination queues are full. 

	\item Nodes $n_2$ and $n_3$ respectively consist of $K_2$ and $K_3$ identical servers
	working in parallel. The servers in $n_2$ have 
	geometrically distributed service times (with service probability $q_2$)
	whereas those in $n_3$ are deterministic, with a service time of $T_3$ slots.

	\item The servers in $T_3$ are prone to faults.
	The probability of a job turning out faulty (denoted $\beta$)
	is inversely related to the service time ($\beta=1/T_3$).
	A job that has received faulty service is sent back to node $n_1$ 
	to be re-processed as a fresh job. If the 
	destination queue at $n_1$ is full, the corresponding
	server in $n_3$ stalls.
	\end{denseItemize}

	The parameter set for this system is $\{C_1,C_2,C_3,T_1,T_3,K_2,K_3\}$.
	Each parameter can take integer values between $1$ and $10$.
	The arrival probability $p$ and the service parameter $q_2$
	are kept fixed ($p=0.5$, $q_2=0.1$).
	Let $X$ denote the vector of parameter values
	and $\Omega_D$ denote the set of all possible values
	that $X$ can take. Let $T(X)$ denote the expected value of the long-run average 
	throughput of the system (estimated using simulation). 
	The throughput can be improved directly by
	improving the value of each parameter (that is, increasing the buffer sizes
	and the number of servers and reducing the service times).
	However, in most real-world applications improving a parameter value
	would also incur a certain cost. To model this trade-off
	we define a synthetic cost function $C(X)$ whose value increases
	with increasing buffer sizes and the number of servers and
	reduces with increasing service-times as follows:
	\begin{equation}\label{eq:cost}
	C(X) =  (C_1+C_2+C_3) + \frac{20}{T_1} + 100  K_2 + 20 \frac{K_3}{T_3}.
	\end{equation}
	The objective function to be minimized is the weighted sum of the 
	normalized cost and throughput components:
	\begin{equation}\label{eq:objective}
	 f(X) = \frac{C(X)}{\underset{j \in\Omega_D}{\max}\ {C(j)}} - \frac{T(X)}{p}
	\end{equation}
	Since an exhaustive evaluation of $f$ over all $10^7$ design points
	is infeasible, our goal is to find the best solution possible
	within a fixed computational budget.
	
	\subsection{Randomization-based Embedding}

	To solve this optimization problem, we first embed the discrete
	parameter space into a continuous one by
	randomizing each parameter in the model 
	using the scheme described in Section \ref{sec:ApplicationToQueues}.
	The randomization settings (listed in Table \ref{tbl:randomization_settings})
	were chosen via a rough tuning of the interpolation curves along each parameter axis.
	\begin{table}[h]
	\centering
	\renewcommand{\arraystretch}{0.8}
	\begin{tabular}{|c c||c c||c c|}
		\hline
		Parameter & $r$, $s$ & Parameter & $r$, $s$ & Parameter & $r$, $s$\\\hline
		$C_1$ & 1, -2 & $T_1$ & 1, 1 & $K_2$ & 1, 4\\\hline
		$C_2$ & 1,  1 & $T_3$ & 1, 1 & $K_3$ & 1, 1\\\hline
		$C_3$ & 1, -2 &  &  &  & \\\hline
	\end{tabular}	
	\caption{Randomization settings (stencil size = 2 for all parameters)}
	\label{tbl:randomization_settings}
	\end{table}
	Figure \ref{fig:ObjectiveContours} shows the resulting interpolation
	plotted along arbitrary two-dimensional slices of the 7-dimensional 
	parameter space.
	The plots were obtained by sweeping two parameters at a time (in steps of $0.25$) 
	while keeping the other parameter values fixed. 
	The plots indicate that in the observed regions of the domain,
	the interpolation is reasonably smooth
	and suited to the application of descent-based optimization methods.
	\begin{figure}[ptbh]
		
		\begin{subfigure}{0.5\textwidth}
		\centering
		\includegraphics[width=1\textwidth]{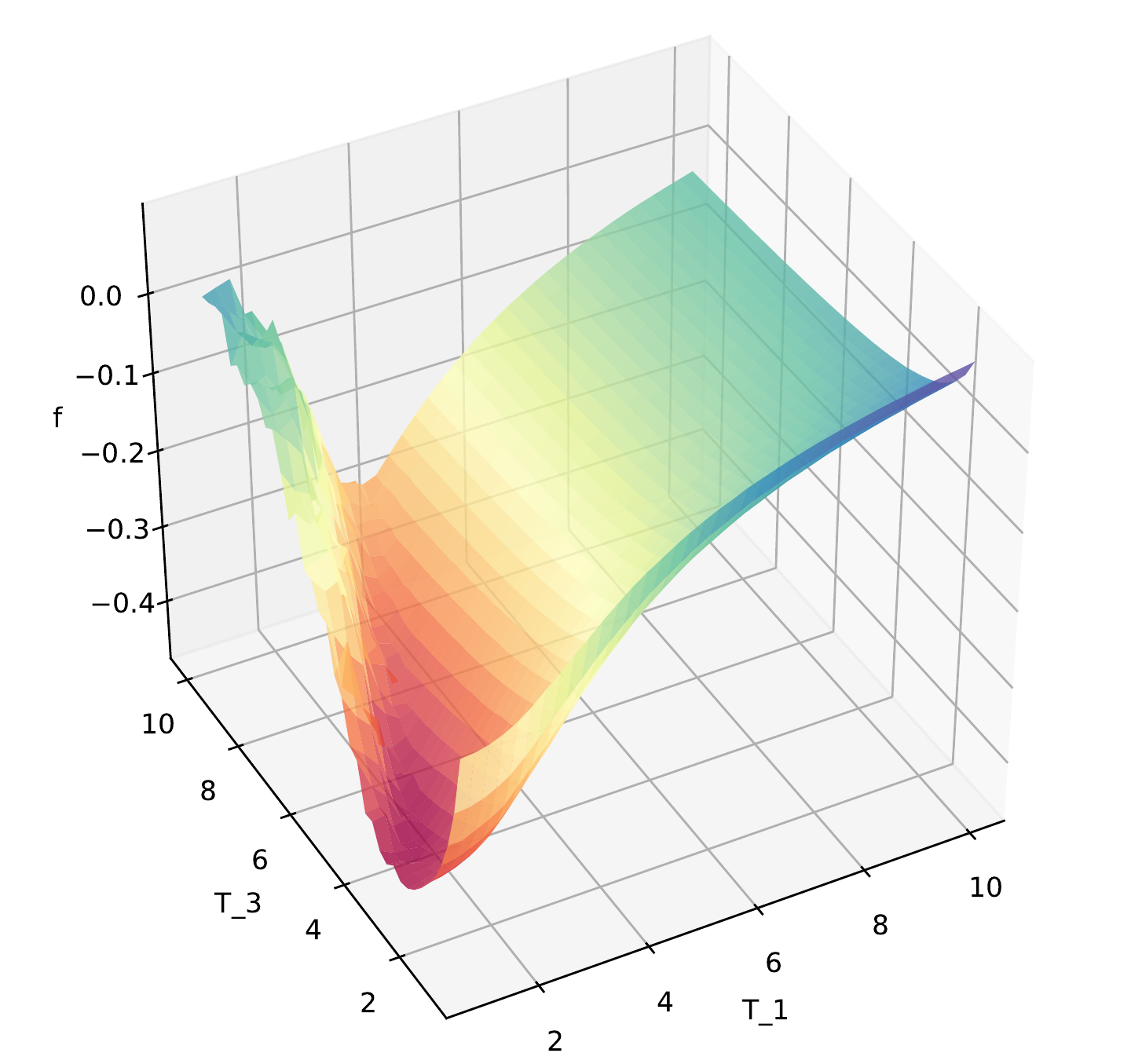}
		\caption{$f(T1,T3)$}
		\label{fig:Opt_T1_T3}
		\end{subfigure}
		\begin{subfigure}{0.5\textwidth}
		\centering
		\includegraphics[width=1\textwidth]{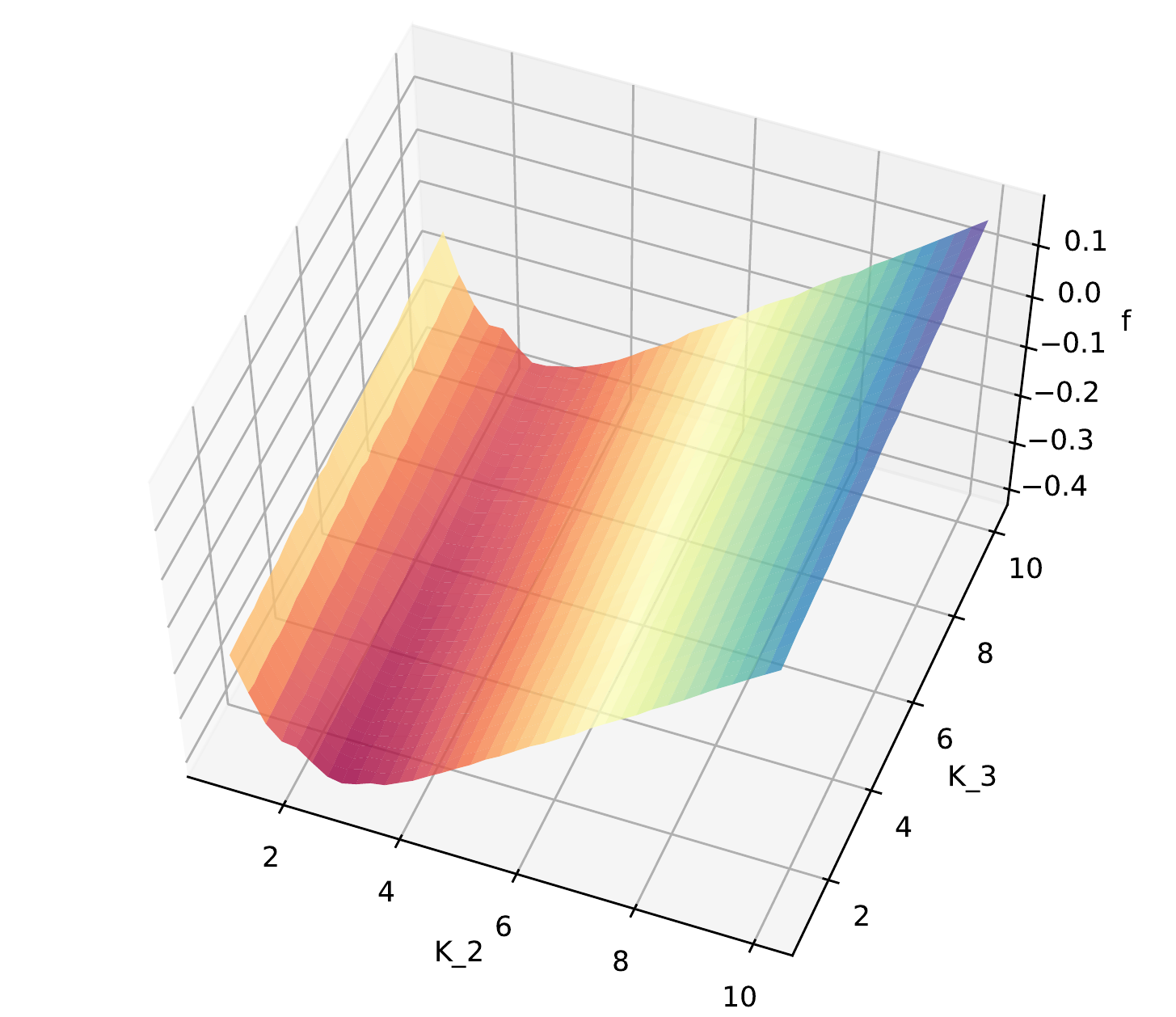}
		\caption{$f(K2,K3)$}
		\label{fig:Opt_K2_K3}
		\end{subfigure}

		\begin{subfigure}{0.5\textwidth}
		\centering
		\includegraphics[width=1\textwidth]{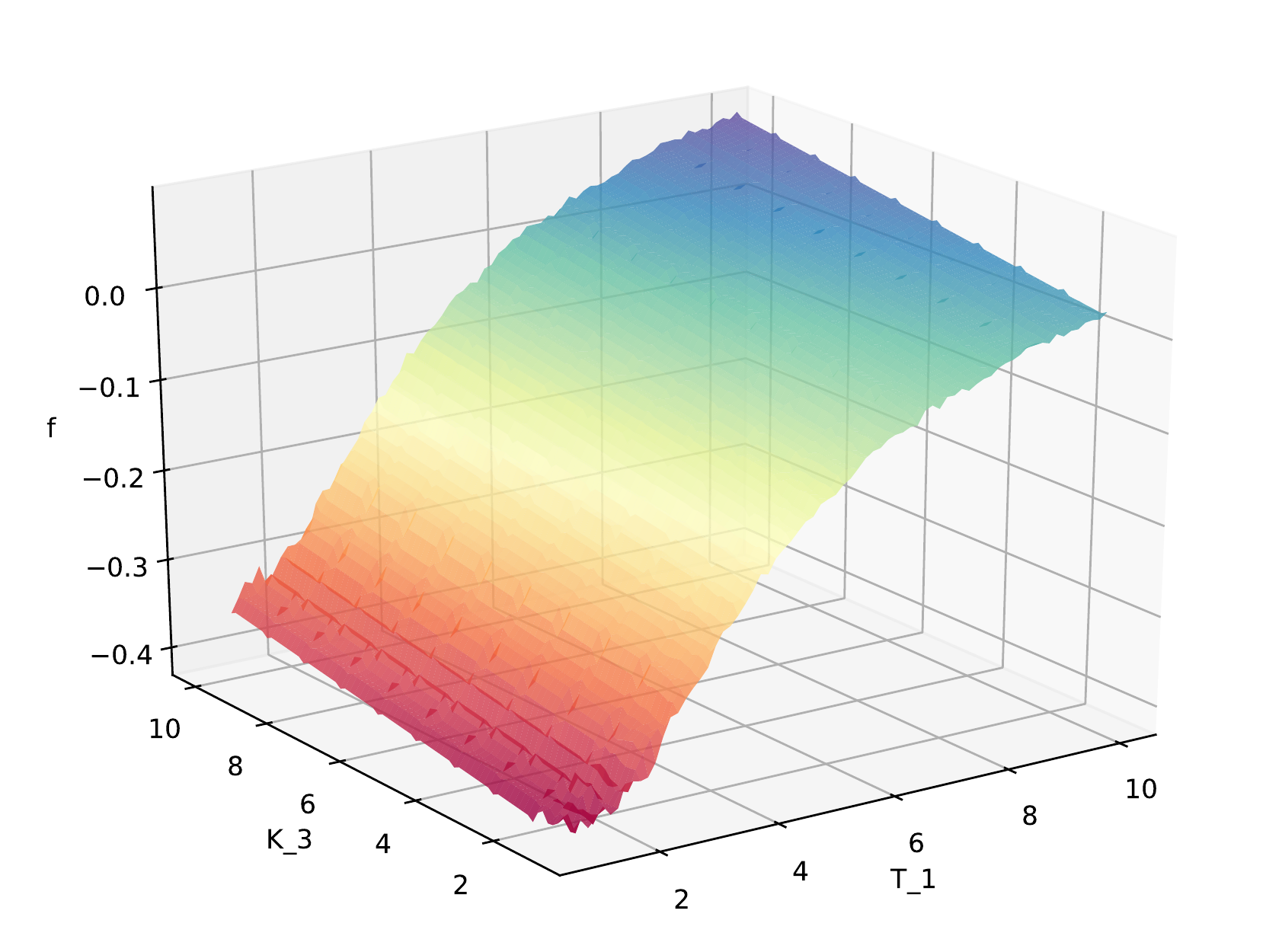}
		\caption{$f(T1,K3)$}
		\label{fig:Opt_T1_K3}
		\end{subfigure}
		\begin{subfigure}{0.5\textwidth}
		\centering
		\includegraphics[width=1\textwidth]{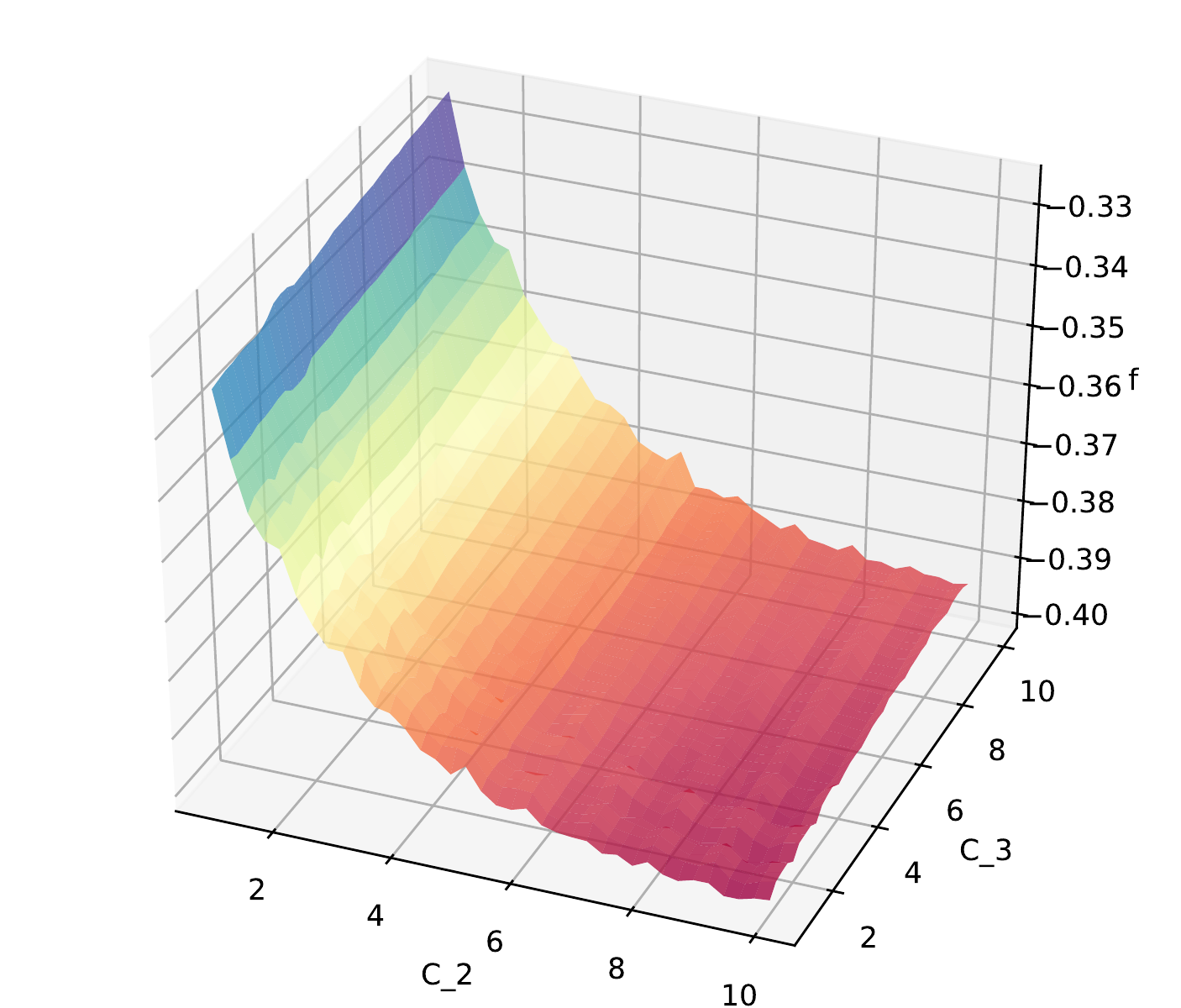}
		\caption{$f(C2,C3)$}
		\label{fig:Opt_C2_C3}
		\end{subfigure}

		\caption{The interpolated objective function 
		plotted along arbitrary 2D slices of the domain.
		Each point on the plot is the objective measured 
		using a single simulation (of length $10^4$ slots) 
		of the randomized model.}
		\label{fig:ObjectiveContours}
	\end{figure}

	In Table \ref{tbl:computational_overhead} we list the 
	computational overheads contributed 
	by the additional calls to the random number generator
	as successive parameters in the model are embedded.
	The overheads were measured relative to the point $X=(5,5,5,5,5,5,5)$
	by successively randomizing each parameter and
	assigning to it the value $5.5$.
	The values in Table \ref{tbl:computational_overhead} 
	show that the overheads are a fraction of the total simulation time.
	This is in contrast to spatial interpolation methods
	where the computational cost of interpolation
	grows as integer multiples of the simulation time.
	\begin{table}[tbh]
	\centering
	\begin{tabular}{|c|r|r|}
	\hline
	
	\textbf{\begin{tabular}[c]{@{}c@{}}Number of\\  parameters \\ embedded\end{tabular}} 
	& \multicolumn{1}{c|}{\textbf{\begin{tabular}[c]{@{}c@{}}Avg real-time\\ per simulation$^*$\\ (seconds)\end{tabular}}} 
	& \multicolumn{1}{c|}{\textbf{Overhead}} \\ \hline
	
	0   & 0.2173   & 0          \\ \hline
	1   & 0.2294   & 5.59 \%    \\ \hline
	2   & 0.2304   & 6.06 \%    \\ \hline
	3   & 0.2302   & 5.96 \%    \\ \hline
	4   & 0.2447   & 12.65 \%   \\ \hline
	5   & 0.2598   & 19.58 \%   \\ \hline
	6   & 0.2701   & 24.31 \%   \\ \hline
	7   & 0.2880   & 32.53 \%   \\ \hline
	\multicolumn{3}{|l|}{\begin{tabular}[c]{@{}l@{}}$^*$ computed across 10 simulation runs (of length $10^6$ slots each)\\ with distinct randomization seeds\end{tabular}} \\ \hline
	\end{tabular}
	\caption{Computational overhead of the randomization-based embedding \\
	for the 7-parameter queueing network in Figure \ref{fig:OptimizationProblem}}
	\label{tbl:computational_overhead}
	\end{table}
		
	\subsection{Optimization}

	Over the randomization-based embedding, we apply 
	two continuous optimizers: 
	COBYLA\cite{Powell1994,COBYLA} 
	and SPSA\cite{Spall1992}.
	Both optimizers avoid the need for an explicit
	computation of numerical derivatives along each parameter axis and are
	thus suited to high dimensional problems. 
	Further, both methods are suited to problems
	where the objective function evaluations can be noisy.
	COBYLA is a trust-region based method.
	It builds an estimate of the gradient at each step by 
	interpolation at the vertices of a simplex in the parameter space.
	SPSA is a descent-based method that 
	approximates the gradient at each step using only
	two objective function evaluations regardless of the dimensionality
	of the parameter space.

	While discrete-parameter variants of SPSA exist\cite{Hill_Discrete_SPSA_2001},\cite[Chapter 9]{Bhatnagar2013},
	COBYLA cannot be applied in the discrete-parameter case.
	This fact illustrates the utility of our embedding technique,
	which makes it possible to apply existing continuous optimizers
	directly to solve discrete-parameter simulation optimization problems.
	We compare the performance of both these optimizers against 
	a simple discrete-space variant of SPSA \cite{Hill_Discrete_SPSA_2001}
	where each objective function evaluation is restricted to 
	points in the original discrete domain.
	While the focus of this work is on the embedding
	technique rather than specific optimization methods,
	we present the comparison as an  
	evaluation of the embedding technique's utility.

	We use an implementation of COBYLA from Python's SciPy
	library\cite{PythonCOBYLA} with the settings $\rho_{\rm beg}=5.0$ and $\rho_{\rm end}=0.1$.
	For SPSA, we use an implementation from the NoisyOpt library\cite{Noisyopt} with the 
	default step-size schedules (suggested in \cite{Spall1998}). The step-size parameter 
	is identical between the continuous and discrete SPSA variants.
	For each method we perform $100$ optimization runs 
	using distinct, randomly chosen initial points. 
	The set of initial points is fixed and is 
	common across all three optimization methods.
	For each optimization run we set an upper limit of $1000$ objective
	evaluations. At each objective function evaluation, the system throughput
	is measured using a single simulation of the randomized model
	(of length $10^4$ slots) and the cost is computed analytically
	using Equation (\ref{eq:cost}).
	At the end of each optimization run, we round the solution
	to the nearest integer point, and record the objective value
	at this point.
	
	%
\begin{table}[h]
\centering
\begin{tabular}{|c|c|c|c|c|}
\hline
\multicolumn{2}{|l|}{}                                                                                                                       & \textbf{COBYLA} & \textbf{SPSA} & \textbf{Discrete-SPSA} \\ \hline
\multirow{3}{*}{\textbf{\begin{tabular}[c]{@{}c@{}}Objective value\\ at the optimum\\ (lower is better)\end{tabular}}}   & \textbf{best}     & -0.7130         & -0.7108       & -0.6842                \\ \cline{2-5} 
                                                                                                                         & \textbf{avg}      & -0.5240         & -0.1994       & -0.1964                \\ \cline{2-5} 
                                                                                                                         & \textbf{std-dev}  & 0.2930          & 0.3481        & 0.3358                 \\ \hline
\multicolumn{2}{|c|}{\textbf{\begin{tabular}[c]{@{}c@{}}Avg number of objective\\ function evaluations\\ per optimization run\end{tabular}}} & 52.7            & 1000          & 1000                   \\ \hline
\multicolumn{2}{|c|}{\textbf{\begin{tabular}[c]{@{}c@{}}Avg time per optimization\\ run (seconds)\end{tabular}}}                             & 0.15            & 2.94          & 2.33                   \\ \hline
\end{tabular}
\caption{Performance of COBYLA and SPSA (applied over a randomization-based embedding) and Discrete-SPSA applied directly over the discrete parameter model.}
\label{tbl:OptimizationResults}
\end{table}

	Table \ref{tbl:OptimizationResults} presents the 
	performance of the three methods
	measured across $100$ optimization runs. 
	We observe that COBYLA shows the best performance, both in terms
	of the quality of the solutions and the convergence rate.
	The continuous-space SPSA performs better in comparison
	to its discrete-parameter variant.
	The solutions are well-clustered. Among the top 
	$20$ solutions found by COBYLA, all have the parameter values
	$T_1=1$, $T_3=10$, and $K_2=3$. The results indicate that 
	existing continuous optimizers can be effectively
	applied over the embedding and their
	performance compares favourably 
	against direct discrete-space search.

\section{Conclusions}

	We have presented a simple and computationally efficient
	technique using which discrete parameters in a simulation optimization problem 
	can be embedded into a continuous-space, 
	enabling direct application of continuous-space
	methods for optimization.
	The technique is based on a randomization of the simulation model
	and is applicable to problems where 
	the objective function is a long-run average measure.
	We described in detail the application of this technique 
	to discrete-time queues for embedding queue capacities,
	number of servers and server-delay parameters
	and empirically showed that the technique can produce 
	smooth interpolations. Further, the interpolation 
	can be tuned by varying the shape of the 
	coefficient functions. An analytical means to 
	arrive at the best randomization scheme (rather than through tuning)
	can be an interesting problem for future research.

	Over the randomization-based embedding, 
	existing continuous optimizers can be effectively
	applied to find good solutions.
	We demonstrated this via 
	an optimization case study
	of a queueing network with $7$ discrete parameters. 
	A randomization of the model produced smooth embeddings
	of the objective function with a low computational
	overhead. Two continuous-space
	optimizers (COBYLA and SPSA) applied over this
	embedding showed good performance in comparison to 
	a direct discrete-space search.
	
	In summary, this work established the randomization
	based technique as an efficient means
	of embedding discrete parameter queueing systems
	into continuous space for simulation-based optimization
	over a large number of parameters.
	A detailed comparison across multiple continuous optimizers
	that can be applied over such an embedding could be 
	the subject of future investigation.
	The application of the embedding technique to other
	kinds of systems (such as inventory models) also needs 
	to be investigated further.

	\bibliographystyle{IEEEtran}
\bibliography{IEEEabrv,bibliography_list}
\end{document}